\documentclass[11pt]{article}
\usepackage{amsmath}
\usepackage{amssymb}
\usepackage{amsthm}

\usepackage{enumerate}

\usepackage{hyperref}

\usepackage{pictexwd,dcpic}

\usepackage[margin=1in]{geometry}

\newcommand{\mc}[1]{\ensuremath{\mathcal{#1}}}

\newcommand{\mb}[1]{\ensuremath{\mathbb{#1}}}
\newcommand{\ra}{\rightarrow}

\newcommand{\N}{\mb{N}}
\newcommand{\Z}{\mb{Z}}
\newcommand{\Q}{\mb{Q}}
\newcommand{\R}{\mb{R}}
\newcommand{\C}{\mb{C}}
\newcommand{\F}{\mb{F}}

\usepackage{tikz}
\usetikzlibrary{matrix,arrows}
\usetikzlibrary{decorations.markings}

\newcommand{\xbinom}{\genfrac{\{}{\}}{0pt}{}}
\newcommand{\qbinom}{\genfrac{[}{]}{0pt}{}}
	
\theoremstyle{plain}
\newtheorem{mythm}{Theorem}[section]
\newtheorem{mylemma}[mythm]{Lemma}

\theoremstyle{definition}
\newtheorem{mydef}[mythm]{Definition}

\newtheorem{myrmk}[mythm]{Remark}

\begin{document}
\setcounter{page}{0}
\title{On Decoding Cohen-Haeupler-Schulman Tree Codes}
\author{Anand Kumar Narayanan \\
Laboratoire d'Informatique de Paris 6, Sorbonne Universit\'e, UPMC Campus,\\
anand.narayanan@lip6.fr
\and Matthew Weidner \\
Computer Science Department, Carnegie Mellon University,\\
maweidne@andrew.cmu.edu}
\date{\today}
\maketitle

\vspace{-10pt}
\maketitle
\vspace{-10pt}
\tikzset{->-/.style={decoration={
			markings,
			mark=at position .5 with {\arrow{>}}},postaction={decorate}}}
\begin{abstract}
		 Tree codes, introduced by Schulman \cite{sch93,sch96}, are combinatorial structures essential to coding for interactive communication. An infinite family of tree codes with both rate and distance bounded by positive constants is called asymptotically good. Rate being constant is equivalent to the alphabet size being constant. Schulman proved that there are asymptotically good tree code families using the Lovasz local lemma, yet their explicit construction remains an outstanding open problem. In a major breakthrough, Cohen, Haeupler and Schulman \cite{chs} constructed explicit tree code families with constant  distance, but over an alphabet polylogarithmic in the length. Our main result is a randomized polynomial time decoding algorithm for these codes making novel use of the polynomial method. The number of errors corrected scales roughly as the block length to the three-fourths power, falling short of the constant fraction error correction guaranteed by the constant distance. We further present number theoretic variants of Cohen-Haeupler-Schulman codes, all correcting a constant fraction of errors with polylogarithmic alphabet size. Towards efficiently correcting close to a constant fraction of errors, we propose a speculative convex optimization approach inspired by compressed sensing. 
		  
\end{abstract}

\thispagestyle{empty}

\setcounter{page}{1}

\section{Introduction}
\subsection{Tree Codes}
A binary tree code is a rooted complete binary tree (say of depth $n$) with a labelling of the vertices from a finite alphabet $\Sigma$. This provides a natural way to map a binary string of length at most $n$ to a string over $\Sigma$ of the same length. Namely, the binary string determines a path from the root by choosing a sequence of child vertices, and the string of labels along the path is assigned to the binary string. We will use the term length to refer to the depth of tree codes interchangeably; they mean the same quantity. Think of this encoding as an error correcting code with rate $1/\log_2|\Sigma|$. The tree structure enforces an \textit{online} or causality constraint on the encoding: two binary strings that agree on their first $k$ coordinates have encodings that agree at the first $k$ coordinates. This online property enables the use of tree codes in interactive communication \cite{sch93, sch96}. The error tolerance of a tree code is quantified by its distance $\delta$, defined as the largest $\delta \geq 0$ such that the encodings of any two strings differ in at least a $\delta$-fraction of the coordinates after their first disagreement. Rate and distance are competing quantities, and a family of tree codes is deemed asymptotically good if the tree codes in it are of increasing depth and yet the rate and distance are both lower bounded by nonzero constants. The constant rate condition is equivalent to the alphabet size $|\Sigma|$ being upper bounded by a constant. Schulman \cite{sch93,sch96} proved there is an asymptotically good family of tree codes. Explicit construction of an asymptotically good family of tree codes remains an open problem. Such an explicit construction could perhaps be a first step towards efficiently decodable tree codes, leading to their use in practical interactive communication.

\subsection{Cohen-Haeupler-Schulman Tree Codes and Online Uncertainty}
\noindent In a recent breakthrough, for any target distance a constant $\delta \in (0,1/2)$, Cohen, Haeupler and Schulman (CHS) \cite{chs} devised a family of binary tree codes of increasing depth $n$ with alphabet size polylogarithmic in $n$ and distance at least $\delta$. Our main contribution is a randomized polynomial time decoding algorithm for these codes which corrects roughly $\sqrt{n}$ errors. We fall short of the $\delta n$ error correction guarantee that can be achieved when run time is of no concern. But prior to our work, no nontrivial polynomial time algorithms were known. \\

\noindent We begin with an informal sketch of the CHS construction before outlining our decoding approach. Prime ingredients in the CHS construction are certain integer tree codes. They map integer strings to integer strings of the same length as follows.
Think of an integer string $(z_0,z_1,\ldots,z_{n-1})$ as a set of evaluations at integer indices and associate to it its Newton basis representation $(a_0,a_1,\ldots,a_{n-1})$, namely, $$z_i  = \sum_{j=0}^{n-1}a_j\binom{i}{j}, \mbox{ for all } i\in\{0,1,\ldots, n-1\}.$$  Such coefficients $a_j$ exist, are unique, and are in fact integral, as seen by the inversion formula $$a_j = \sum_{i=0}^{n-1} (-1)^{j-i}\binom{j}{i}z_i.$$ The integer tree code maps $$(z_0,z_1,\ldots,z_{n-1}) \longmapsto ((z_0,a_0),(z_1,a_1),\ldots,(z_{n-1},a_{n-1}))$$ where each pair $(z_i,a_i)$ on the right is thought of as being encoded as a single integer. These encodings conform to a curious distance property: encodings of two sequences of the same length differ (after their first disagreement) in at least half the coordinates. To see why, since the mapping is linear over the integers, it suffices to show separation for the encoding of an arbitrary nonzero string $(z_0,z_1,\ldots,z_{n-1})$ from the all zero string. 
Denote by $s$ the first index of disagreement, that is, $s$ is the smallest index with $z_s\neq 0$. The combinatorial Lindstr\"om-Gessel-Viennot Lemma \cite{lin73,GV85} implies the uncertainty principle $$Sparsity(z_s,z_{s+1},\ldots,z_{n-1}) + Sparsity(a_s,a_{s+1},\ldots,a_{n-1}) \geq n-s+1 $$ for the Newton basis transformation, where $Sparsity$ of a string is the number of nonzero elements. This uncertainty immediately implies the desired distance property $$Sparsity((z_s,a_s),(z_{s+1},a_{s+1}),\ldots,(z_{n-1},a_{n-1})) \geq \frac{n-s}{2}.$$ 
It is remarkable that the uncertainty is additive! Contrast this with the more familiar Heisenberg uncertainty (say, in the context of discrete Fourier transforms) which is multiplicative: the product of the sparsity of a string and its Fourier transform is bigger than the dimension.  Further, the Newton basis transformation is online, i.e., $a_i$ does not depend on $\{z_{k} \mid k>i\}$, since the binomial $\binom{i}{k}$ vanishes for $k>i$. The integer tree code encoding map $$(z_0,z_1,\ldots,z_{n-1}) \longmapsto ((z_0,a_0),(z_1,a_1),\ldots,(z_{n-1},a_{n-1}))$$  thus constitutes an online encoding with distance at least $1/2$. These integer tree codes are then carefully wrapped in intricate classical coding theory machinery to bear binary tree codes. The distance of the resulting tree codes are $1/3$. To obtain tree codes with any prescribed distance $\delta <1$, a modified version of the integer tree codes and wrapping machinery suffice, but all the salient features of the scheme are captured in the tree code construction with distance $1/3$. The alphabet size is determined by how quickly the Newton basis coefficients grow and is polylogarithmic in the depth of the tree code. Curiously, to obtain tree code families having constant distance with even smaller alphabets, it suffices to construct a transformation having additive uncertainty with the transform coefficients smaller than Newton basis coefficients. 

\subsection{Decoding by the Polynomial Method}

\noindent On the decoding front, this classical coding theory machinery is easy to unwrap, meaning the task of decoding the binary tree codes reduces to that of decoding the integer tree codes. 
Call the message string $z:=(z_0,z_1,\ldots, z_{n-1})$ of the integer tree code and its Newton basis coefficients $a:=(a_0,a_1,\ldots, a_{n-1})$. The decoder is given an erroneous version $((\hat{z}_0,\hat{a}_0),(\hat{z}_1,\hat{a}_1),\ldots,(\hat{z}_{n-1},\hat{a}_{n-1}))$ of the encoding and has to recover $z$ (or equivalently $a$). Assume the number of errors (that is, the number of pairs with $(\hat{z}_i,\hat{a}_i) \neq (z_i,a_i)$), possibly adversarial, is bounded by $\mathcal{O}(\sqrt{n/\log(n)})$.\\

\noindent  Write the corrupted Newton coefficients as $\hat{a}=a+v$, where $v$ is the (sparse) Newton basis error vector, and write the corrupted evaluations as $\hat{z} = z + u$, where $u$ is the (sparse) evaluation basis error vector.  In light of the linearity of the encoding, we have
\[
\sum_{j=0}^{n-1} v_j\binom{i}{j} = u_i + \left(\sum_{j=0}^{n-1} \hat{a}_j \binom{i}{j} - \hat{z}_i \right), \mbox{ for all } i \in \{0, 1, \dots, n-1\}.
\]
Thus our goal is to interpolate a sparse polynomial in the Newton basis, given an erroneous version of its evaluations. \\

\noindent \textbf{Locating Evaluation Non-errors:} It is convenient to distinguish between errors in the $z_i$ values (call them evaluation errors) and errors in the $a_i$ values (call them Newton errors). The first step in our algorithm is to locate sufficiently many correct evaluations. Specifically, for some $\alpha(n)$, we find a set $I_z$ of $\alpha(n)/2$ indices among the first $\alpha(n)$ indices such that the evaluations are correct at all indices in $I_z$. We locate these indices one at a time, each time using a polynomial time algebraic algorithm relying on the polynomial method. Our algebraic algorithm is a Newton basis analogue of the Sudan \cite{sud97} and Shokrollahi-Wasserman \cite{SW99} algorithms, but with careful randomization in the interpolation step.  \\

\noindent \textbf{Locating Newton Non-errors using Duality:} The next step is to locate sufficiently many correct Newton coefficients.  That is, we want to find a set $I_a$ of $\alpha(n)/2$ indices among the first $\alpha(n)$ indices such that the Newton coefficients are correct at all indices in $I_a$. To this end, we identify a duality for the Newton basis transformation apparent from the inversion formula:
the Newton basis transformation matrix is a certain row and column scaled version of its own inverse. A consequence of this duality is that we can translate the Newton basis non-error location problem into an evaluation basis non-error location problem, solving the latter using the aforementioned polynomial time algorithm. \\ 

\noindent \textbf{Non-error Location to Error Correction:} We are left with two sets, $I_z$ containing $\alpha(n)/2$ correct evaluation indices and $I_a$ containing $\alpha(n)/2$ correct Newton basis error indices among the first $\alpha(n)$ indices. The uncertainty principle from the Lindstr\"om-Gessel-Viennot Lemma ensures that $z_0$ can be extracted from these correct values using linear algebra.  We can then subtract the encoding of $(z_0, 0, \dots, 0)$, delete the first coordinate, and repeat the algorithm to find $z_1$, then $z_2$, etc., yielding the sought string $(z_0,z_1,\ldots,z_{n-1})$ in polynomial time. 

\begin{mythm}\label{thm_intro_integer_decoding}
	There is a function $\epsilon(n) = \Omega(\sqrt{n/\log(n)})$ for which the Las Vegas type randomized algorithm in \S~\ref{section_alg_alg} has the following guarantee: Given an input $((\hat{z}_0,\hat{a}_0),(\hat{z}_1,\hat{a}_1),$ $\ldots,(\hat{z}_{n-1},\hat{a}_{n-1}))$, where all $\hat{z}_i \in [-Z, Z]$ and $\hat{a}_i \in [-Z2^n, Z2^n]$, the algorithm in expected runtime polynomial in $n$ and $\log(Z)$ outputs a $z = (z_0, \dots, z_{n-1}) \in [-Z, Z]^n$ such that the encoding of $z$ and the input have Hamming distance at most $\epsilon(n)$, if such a $z$ exists.
\end{mythm}

\noindent From this decoding algorithm for the integer tree codes developed in \S~\ref{section_alg_alg}, we obtain a decoding algorithm in \S~\ref{alphabet_reduction_section} for the binary tree codes correcting $\Omega\left(n^{3/4}\log^{-1/2}(n)\right)$ errors.

\begin{mythm}\label{thm_intro_binary_decoding}
	There is a Las Vegas type randomized decoding algorithm (see definition \ref{definition_decoding}) for the Cohen-Haeupler-Schulman tree codes of length $n$ correcting $\Omega\left(n^{3/4}\log^{-1/2}(n)\right)$ errors in expected polynomial time.
\end{mythm} 

\subsection{Decoding by Convex Optimization} 
The integer tree code decoding problem of finding a sparse polynomial (coefficient vector) $v$ in the Newton basis may be recast as a combinatorial optimization problem 
$$\arg\min_{u,v}\left(\left|\left|u\right|\right|_{\ell_0} + \left|\left|v\right|\right|_{\ell_0}\right), \mbox{ such that}\ \hat{z}-u=B(\hat{a}-v),$$ where $B$ denotes the lower triangular matrix $\{\binom{i}{j} \mid i,j\in [0,n)\}$ with binomial coefficients. We suggest relaxing and solving instead the convex optimization problem $$\arg\min_{u,v}\left(\left|\left|u\right|\right|_{\ell_1} + \left|\left|v\right|\right|_{\ell_1}\right), \mbox{ such that}\ \hat{z}-u=B(\hat{a}-v),$$ where the $\ell_0$ norm is replaced with the $\ell_1$ norm. This proposal is inspired by Compressed Sensing, where the recovery of sparse vectors from only a few observations under certain linear transformations is accomplished by convex optimization. The success of the optimization to recover the sparse signal in compressed sensing is guaranteed by a Restricted Isometry Property (RIP) of the linear transformation.\\

\noindent We formulate an analogous Online Restricted Isometry Property (Online RIP) as a sufficient condition for the decoding to succeed in most cases. If the binomial matrix $B$ were to satisfy this property, we would have an algorithmic online uncertainty principle (for $z_0 \neq 0$)
$$Sparsity(z_0,z_1,\ldots,z_{n-1}) + Sparsity(a_0,a_1,\ldots,a_{n-1}) \geq \frac{n+1}{(\log n)^{\mathcal{O}(1)}}.$$
The resulting uncertainty is weaker in that there is a small polylogarithmic loss on the right. However, it is algorithmic: the convex optimization corrects $n/(\log n)^{\mc{O}(1)}$ errors, nearly that guaranteed by the constant distance property. Further, we only need to work with a precision that allows for the convex optimization, perhaps a means for reducing the alphabet size. \\ 

\noindent This leaves open the problem of constructing matrices satisfying Online RIP. Several random matrix ensembles are known to satisfy RIP, but none are online. A thorough investigation of the (randomized or explicit) construction of matrices satisfying Online RIP is hence warranted. We present explicit candidate constructions that are number theoretic twists on the Newton basis transformation matrix.

\subsection{Number theoretic extensions of Cohen-Haeupler-Schulman Codes}  
We extend the Cohen-Haeupler-Schulman framework to obtain binary tree code families with similar properties, namely, positive constant distance and polylogarithmic alphabet size. Again, the main ingredients are integer tree codes wrapped in classical coding theory machinery.  \\

\noindent Our first two constructions look to Gaussian or $q$-binomial coefficients   

\[ \qbinom{r}{s}_q :=  \begin{cases} 
\frac{(q^r-1)(q^{r-1}-1)\ldots(q^{r-s+1}-1)}{(q^s-1)(q^{s-1}-1)\ldots(q-1)} & , r \geq s \\
\ \ \ \ 0 & , r < s 
\end{cases}
\] for nonnegative integers $r,s$. Despite appearances, $\qbinom{r}{s}_q$ is indeed a polynomial in $q$ for $r \geq s$.  The online uncertainty is then derived from the evaluation map $$z_i  = \sum_{j=0}^{n-1}a_j\qbinom{i}{j}_{q} , i \in \{0,1,\ldots,n-1\}$$ 
after substituting for $q$ a carefully chosen complex number on the unit circle.\\

\noindent \textbf{Tree Codes from Cyclotomic Units:} In the first construction, $q$ is chosen to be a primitive $\ell^{th}$ root of unity for a large enough $\ell = \Theta(n^3)$. Curiously, all the nontrivial entries of the matrix of transformation are cyclotomic units. For the transformation to satisfy additive online uncertainty, determinants of certain submatrices of the transformation (which turn out to be algebraic integers) should not vanish. We prove this by calling upon the Lindstr\"om-Gessel-Viennot framework with underlying lattice paths of carefully selected edge weights. These tree code families result in binary tree codes having constant distance with polylogarithmic alphabet size.\\ 

\noindent \textbf{Tree Codes from Sunflowers:} The second construction is inspired by phyllotaxis, the arrangement of leaves on the stems of certain plants. The same phenomenon arranges sunflower petals. The idea is to substitute for $q$ a transcendental number of the form $e^{2\pi i \alpha}$ where $\alpha$ is an irrational algebraic real number. Integer multiples  $\{i\alpha \mod 1 \mid 0 \leq i \leq t-1\}$ of $\alpha$ modulo $1$ distribute on the unit circle in a pattern determined by the three gap theorem. For any $t$, there are only three gaps between the nearest points, with the sum of the two smaller gaps equalling the third. The smallest gap is maximized (scales as $\Theta(1/t)$) when $\alpha$ is the golden section $(\sqrt{5}-1)/2$, a choice common in plants including sunflower petals! This may be seen as a pseudorandom way of populating the unit circle while still maintaining separation.\\ 

\noindent Again, we look to the Lindstr\"om-Gessel-Viennot framework to prove online uncertainty. In this case, the nonvanishing of the determinants is not difficult to establish, since the determinant is an algebraic expression involving a transcendental $e^{2\pi i \alpha}$. However to construct codes over finite alphabets, we need to represent the transcendental numbers with finite precision approximation. The difficulty now is to show that the determinants do not vanish in small precision. Using effective versions of the Lindemann-Weierstrass theorem from transcendental number theory, we show that a small enough precision suffices. The resulting binary tree code families are of constant distance over polylogarithmic size alphabets. \\  

\noindent \textbf{Tree Codes from the Unit Circle:} We saw that the distribution of points  $\{i\alpha \mod 1 \mid 0 \leq i \leq t-1\}$ on the unit circle is structured,  satisfying the three gap theorem. This structure is an obstruction to these set of points approaching the uniform distribution on the unit circle. Weyl showed that integer square multiples $\{i^2\alpha \mod 1 \mid 0 \leq i \leq t-1\}$ converge to the uniform distribution for any irrational algebraic real $\alpha$. Rudnick, Sarnak and Zaherescu related the convergence to the Diaphantine approximability of $\alpha$. Assuming the ABC conjecture, they proved a strong convergence theorem.  \\

\noindent To construct tree codes from the pseudorandom distribution $\{i^2\alpha \mod 1 \mid 0 \leq i \leq t-1\}$, we move from $q$-binomials to certain multinomials. The online uncertainty is again shown using the Lindstr\"om-Gessel-Viennot framework to accomodate a generalization of $q$-binomials along with an effective Lindemann-Weierstrass theorem.

\subsection{Related Work}
Despite their elegance and applications to interactive communication, relatively few works have studied explicit constructions of tree codes, and even fewer have studied nontrivial decoding algorithms.  Schulman's original proofs \cite{sch93,sch96} that there exist asymptotically good families of binary tree codes are not explicit, relying on the probabilistic method.  Braverman \cite{bra12} gives an asymptotically good construction with subexponential time encoding and decoding algorithms.  Evans, Klugerman and Schulman \cite{sch94} give an explicit construction of binary tree codes with constant distance and alphabet size polynomial in the length; these are easily decoded in polynomial time.  Gelles et~al.\ \cite{gelles_et_al} give efficiently encodable and decodable tree codes with rate approaching one over polynomial size input and alphabets; like our approach, they do not correct a constant fraction of errors, instead only correcting $\Omega(n/\log(n))$ errors.  Unlike the CHS tree codes, their codes' true distance is only known to be $\Omega(1/\log(n))$, not $\Omega(1)$.  Moore and Schulman \cite{ms14} give a candidate construction relying on a conjecture about exponential sums, but nothing is known about decoding these codes.  Pudl\'ak \cite{pud16} gives a construction with large input and output alphabets, using a technique which we reuse in \S \ref{section_cyclotomy}--\ref{section_transcendence}, namely, construct an online matrix depending on some indeterminate, write the determinants of online submatrices of this matrix as polynomials in the indeterminate, then substitute a value which cannot be a root of any of these polynomials.  Again, nothing is known about decoding.

\subsection{Organization}
\S \ref{section_tree_codes} gives background on tree codes and explains the construction of CHS integer tree codes.  \S \ref{alphabet_reduction_section} explains the CHS classical coding theory wrapper yielding binary tree codes from integer tree codes, as well as how to use our decoding algorithm for the integer tree codes to get a polynomial time decoding algorithm for CHS binary tree codes correcting $\Omega\left(n^{3/4}\log^{-1/2}(n)\right)$ errors.  \S \ref{section_alg_alg} presents our algebraic decoding algorithm for the integer tree codes.  \S \ref{convex_decoding_section} discusses a potential alternate approach to tree code decoding relying on convex optimization.  Finally, \S \ref{section_cyclotomy}--\ref{section_transcendence} present new variants of CHS integer tree codes as candidates for convex optimization decoding.  These codes are constructed from q-binomial and multinomial matrices using cyclotomic units and transcendental numbers.

\section{Tree Codes and Online Uncertainty}\label{section_tree_codes}
We formally define tree codes and recount the online uncertainty principle of Cohen-Haeupler-Schulman that is the foundation of their codes.
\subsection{Tree Code Preliminaries}
Recall that binary tree codes are rooted complete binary trees with a labelling of the vertices from a finite alphabet $\Sigma$, encoding binary strings to $\Sigma$ strings with a certain distance property. It is convenient to abandon the tree notation and adopt a functional notation to consider more general tree codes. These general tree codes allow the input strings to be nonbinary, and we will use $\Sigma_1$ and $\Sigma_2$ to distinguish between input and output alphabets. Furthermore, these alphabets could be finite or infinite.  We use the notations $[a] := \{1, 2, \dots, a\}$, $[a, b] := \{a, a + 1, \dots, b\}$, and $[a, b) := \{a, a+1, \dots, b - 1\}$. Strings are indexed starting with $1$, and for a string $x$ and positive integers $i\leq j$, $x_i$ denotes the $i^{th}$ element of the string and $x_{[i,j]}$ the substring from index $i$ through $j$. For two strings $x,y$, let $split(x,y)$ denote the first index where they differ.

\begin{mydef}
	(\textit{Online Functions}) For a positive integer $n$ and alphabet sets $\Sigma_1$ and $\Sigma_2$, a function $f: \Sigma_1^n \longrightarrow \Sigma_2^n$ is said to be online if for all $i \in [n]$ and for all $x \in \Sigma_1^n$, $f(x)_i$ is determined by $x_{[1,i]}$.
\end{mydef}

\begin{mydef}
	(\textit{Tree Codes}) For a positive integer $n$, alphabet sets $\Sigma_1$ and $\Sigma_2$ and distance $0<\delta<1$, a function $TC : \Sigma_1^n \longrightarrow \Sigma_2^n$ is a tree code with distance $\delta$ if 
	\begin{itemize}
		\item $TC$ is online and 
		\item for every distinct $x,y \in \Sigma_1^n$ with split $s:=split(x,y)$ and every $\ell \in [0, n-s]$, $$d_H\left(TC(x)_{[s,s+\ell]},TC(y)_{[s,s+\ell]}\right) \geq \delta(\ell+1)$$ where $d_H$ denotes the Hamming distance. 
	\end{itemize}	 
\end{mydef}

\noindent When $\Sigma_1$ and $\Sigma_2$ are both rings, a tree code $TC:\Sigma_1^n\longrightarrow\Sigma_2^n$ is linear if the function $TC$ is linear.\\

\noindent \textbf{Binary Tree Codes, Integer Tree Codes etc.:} The term ``binary tree codes'' refers to $TC:\Sigma_1^n\longrightarrow\Sigma_2^n$ with $\Sigma_1 =\{0,1\}$. 
Likewise, the term ``integer tree codes'' refers to $TC:\Sigma_1^n\longrightarrow\Sigma_2^n$ with $\Sigma_1 = \Sigma_2 = \Z$. From here on, we will generally use the $TC:\Sigma_1^n\longrightarrow\Sigma_2^n$ notation to prevent ambiguity. There are also versions of tree codes with index set infinite in the literature, say of the form $TC:\Sigma_1^\N\longrightarrow\Sigma_2^\N$, but we refrain from dealing with them for notational convenience.\\

\subsection{Lindstr\"om-Gessel-Viennot Lemma and Newton Basis Uncertainty}
We begin by describing the Cohen-Haeupler-Schulman construction of the integer tree codes $TC_{\Z}$ with distance $1/2$ sketched in the introduction.
Associate to an integer sequence $(z_0,z_1,\ldots,z_{n-1})$\footnote{When dealing with inputs and outputs to $TC_\Z$, we use 0-indexed vectors, to maintain consistency with the notation in the original paper.  Otherwise, we use 1-indexed vectors.} (thought of as an evaluation basis representation) its Newton basis representation $(a_0,a_1,\ldots,a_{n-1})$, defined by the relations $$z_i  = \sum_{j=0}^{n-1}a_j\binom{i}{j}, \mbox{ for all } i\in[0,n).$$ The coefficients $a_j$ are in fact integers, as seen by the inversion formula $$a_j = \sum_{i=0}^{n-1} (-1)^{j-i}\binom{j}{i}z_i.$$ As evident, the transformations between the evaluation and Newton bases in both directions are given by lower triangular integer matrices. The integer tree code $TC_{\Z}$ of length $n$ is defined by
\begin{align*}
TC_{\Z}:\Z^n &\longrightarrow \Z^n \\(z_0,z_1,\ldots,z_{n-1}) &\longmapsto ((z_0,a_0),(z_1,a_1),\ldots,(z_{n-1},a_{n-1})),
\end{align*}
where each pair $(z_i,a_i)$ on the right is encoded as a single integer. Cohen, Haeupler and Schulman used the Lindstr\"om-Gessel-Viennot Lemma to prove that $TC_{\Z}$ has distance $1/2$. While the proof is not important to the description of our decoding algorithms, we take time to sketch it. The proof framework will be critical to the construction of our number theoretic codes. Those interested only in decoding may take for granted Lemma \ref{gessel_viennot} and skip the rest of this subsection. \\

\noindent The Lindstr\"om-Gessel-Viennot Lemma is a famous combinatorial result on determinants of matrices arising from path graphs. There are several excellent expositions including \cite{GV85} and \cite[Chapter 29]{az}. Consider a directed acyclic graph $G=(V,E)$ with edge weights $\{w(e)\mid e\in E\}$ coming from a commutative ring with identity, along with two ordered vertex sets $R=\{r_1,r_2,\ldots,r_d\},C=\{c_1,c_2,\ldots,c_d\} \subseteq V$ of the same cardinality $d$. Associated to it is the path matrix $M$; this is a square matrix with the $r\in R,c\in C$ entry $$M_{r,c}:=\prod_{P:r\rightarrow c}w(P)$$
where the product is taken over all paths $P$ from $r$ to $c$ and the weight $w(P)$ is the product of edge weights in the path $P$. Paths of length $0$ are included and given the weight $1$.  A path system $\mathcal{P}$ from $R$ to $C$ consists of a permutation $\sigma \in S_d$ and a set of paths $\{P_i: r_i \rightarrow c_{\sigma(i)} \mid i \in [d]\}$. Let $sgn(\mathcal{P})$ denote the sign of $\sigma$ and $w(\mathcal{P})$ denote the product of the weights $\prod_{i=1}^d w(P_i)$. Further, the path system is called vertex disjoint if its set of paths are vertex disjoint.\\ 

\noindent The Lindstr\"om-Gessel-Viennot Lemma is the expression for the determinant of the path matrix $M$ in terms of the underlying path graph $$\det(M) = \sum_{\substack{vertex\ disjoint\\ path\ systems\ \mathcal{P}}} sgn(\mathcal{P}) w(\mathcal{P}).$$  Gessel and Viennot applied it to path graphs cut out from the square lattice and proved the following nonvanishing theorem for determinants of Pascal submatrices.

\begin{mylemma}\label{gessel_viennot}
	For nonnegative integers $r_1 < r_2 <\ldots < r_d$ and $c_1<c_2<\ldots<c_d$ with $r_i \geq c_i$ for all $i \in \{1,2,\ldots,d\}$, the determinant of the matrix $\{\binom{r_i}{c_j} \mid i, j \in [d] \}$ is nonzero.  
\end{mylemma}
\begin{proof}
Consider the directed acyclic graph below with unit edge weights and distinguished vertex subsets $\{r_1,r_2,\ldots,r_d\}$ and $\{c_1,c_2,\ldots,c_d\}$.  The $r_1^{th}$ vertex on the first column is labelled $r_1$, the $r_2^{th}$ vertex on the first column is labelled $r_2$ and so on. Likewise, the $c_1^{th}$ vertex on the diagonal is labelled $c_1$, the $c_2^{th}$ vertex on the diagonal is labelled $c_2$ and so on. The horizontal edges are directed from left to right and the vertical edges from bottom to top. All edges have weight $1$.\\

\begin{tikzpicture}
\draw[fill] (0,8) circle (1pt);
\draw[dashed,->-](0,7)--(0,8);
\draw[->-] (0,7) -- (0,8);
\node at (0, 8.3)  {$0$};

\foreach \y in {1,2,3,4,5,6,7,8}
\foreach \x in {0,...,\y}
\draw[fill] (\x,8-\y) circle (1pt) coordinate (m,\x,\y);

\foreach \y in {1,2,3}
\foreach \x in {0,...,\y}
\draw[->-](\x,7-\y) -- (\x,8-\y);

\foreach \x in {0,1,2,3,4}
\draw[dashed,->-](\x,3)--(\x,4);

\foreach \y in {5,6}
\foreach \x in {0,...,\y}
\draw[->-](\x,7-\y) -- (\x,8-\y);

\foreach \x in {0,1,2,3,4,5,6,7}
\draw[dashed,->-](\x,0)--(\x,1);

\foreach \y in {8}
\foreach \x in {0,...,\y}
\draw[->-](\x,7-\y) -- (\x,8-\y);

\foreach \y in {1,2,3,4,5,6,7,8}
\foreach \x in {1,...,\y}
\draw[->-](\x-1,8-\y) -- (\x,8-\y);

\node at (-0.5, 7)  {$1$}; 
\node at (-0.5, 6)  {$2$}; 
\node at (-0.5, 5)  {$r_1$}; 
\node at (-0.5, 4)  {$r_2$}; 
\node at (-0.5, 0)  {$r_i$}; 

\node at (1.5,7) {$1$}; 
\node at (2.5,6) {$c_1$}; 
\node at (3.5,5) {$c_2$}; 
\node at (4.5,4) {$4$}; 
\node at (5.5,3) {$c_j$};

\node at (10,1.5) {$r_i$};
\node at (11.5,4) {$.$};
\node at (13.2,4.2) {$c_j$};
\node at (13,2.5) {$.$};
\draw[->-](11.5,4) -- (13,4) node[midway,above]{\textcolor{red}{$1$}};
\draw[->-](13,2.5) -- (13,4) node[midway,sloped,left,rotate=270]{\textcolor{red}{$1$}};
\draw[dashed,->-](10.2,1.7) to[bend left] (11.5,4);
\draw[dashed,->-](10.2,1.7) to[bend right] (13,2.5);

\node at (9.8,3.5) {$P_{r_i-1,c_j-1}$};

\node at (12.4,1.4) {$P_{r_i-1,c_j}$};

\node at (10.5,5.5) {Number of $r_i \ra c_j$ paths $P_{r_i,c_j} = P_{r_i-1,c_j-1}+P_{r_i-1,c_j}$};

\end{tikzpicture}\\

\noindent Since all the edge weights are 1, the $(i,j)^{th}$ entry $M_{i,j}$ of the path matrix is the number of paths $P_{r_i,c_j}$ from $r_i$ to $c_j$. This satisfies the two term recurrence $P_{r_i,c_j} = P_{r_i-1,c_j-1}+P_{r_i-1,c_j}$ as evident from the picture on the right. This is Pascal's identity for binomials. The boundary conditions $\binom{r_i}{0} =1$ and $\binom{r_i}{c_i}=1$ for $r_i=c_i$ are consistent with the path formulation. We conclude that the associated path matrix is indeed $\{\binom{r_i}{c_j}  \mid i, j \in [d] \}$. The geometry forces all vertex disjoint path systems to have the identity permutation, which has sign $1$. Hence the determinant is a positive number provided there is at least one vertex disjoint path system. By the condition $r_i \ge c_i$ for all $i$, there is at least one, namely for each $r_i \ra c_i$, traverse $c_i$ edges right before turning up.  
\end{proof} 
\noindent Cohen, Haeupler and Schulman derived the following additive uncertainty from this nonvanishing.
\begin{mylemma}\label{newton_uncertainty_lemma}
	Let $(z_0,z_1,\ldots,z_{n-1})$ be an integer string with Newton basis representation $(a_0,a_1,\ldots,$ $a_{n-1})$. If $c$ is the first index with $z_c$ nonzero, then $$Sparsity((z_c,z_{c+1},\ldots,z_{n-1}))+Sparsity((a_c,a_{c+1},\ldots,a_{n-1}))\geq n-c + 1.$$ Equivalently, the polynomial $\sum_{j=0}^{n-1}a_j \binom{x}{j}$ of sparsity $s$, whose first nonvanishing when evaluated on nonnegative integers is $c$, has at most $s-1$ integer roots greater than $c$. 
\end{mylemma}

\section{Alphabet Reduction Machinery}\label{alphabet_reduction_section}
The tree code $TC_\Z$ discussed in the previous section is an integer tree code, but we are ultimately interested in binary tree codes.  To this end, Cohen, Haeupler and Schulman created binary tree codes by wrapping $TC_\Z$ in intricate classical coding theory machinery \cite[\S 6]{chs}.\\

\noindent As the first step in this construction, for any $n \ge 1$, by restricting the input alphabet to $[-Z, Z]$ and applying the bound $|TC_\Z(z_0, \dots, z_{i-1})_i| \le 2^i\max\{z_0^2, \dots, z_{i-1}^2\}$ (\cite[Theorem 1.3]{chs}), we get a tree code
\[
[-Z, Z]^n \ra [-2^nZ, 2^nZ]^n
\]
with distance $1/2$, i.e., a tree code having input alphabet $[-Z, Z]$, output alphabet $[-Z2^n, Z2^n]$, and length $n$.  These are the tree codes that we decode in \S \ref{section_alg_alg}.\\

\noindent \textbf{Construction of the binary tree codes:} We now briefly sketch the remainder of the CHS binary tree code construction and explain how to decode it, given a decoding algorithm for $TC_\Z$.  A reader interested only in decoding integer tree codes may skip this section.\\

\noindent The construction proceeds in several steps.\\

\noindent \textbf{Step 0:} For any $\ell \ge 1$, from the above tree code $[-Z, Z]^\ell \ra [-2^\ell Z, 2^\ell Z]^\ell$ with $Z = 2^{\ell-1}$, we get a tree code
\[
TC_{\ell}: \left(\{0, 1\}^\ell\right)^\ell \ra \left(\{0, 1\}^{3\ell}\right)^\ell
\]
with distance $1/2$, i.e., a tree code having input alphabet $\{0, 1\}^\ell$, output alphabet $\{0, 1\}^{3\ell}$, and length $\ell$.\\

\noindent \textbf{Step 1:} This step constructs a modified type of tree code called a lagged tree code.

\begin{mydef}
	(\textit{Lagged Tree Codes}) For a positive integer $n$, alphabet sets $\Sigma_1$ and $\Sigma_2$, distance $0<\delta<1$ and lag $L$, a function $TC : \Sigma_1^n \longrightarrow \Sigma_2^n$ is a lagged tree code with distance $\delta$ and lag $L$ if 
	\begin{itemize}
		\item $TC$ is online and 
		\item for every distinct $x,y \in \Sigma_1$ with split $s:=split(x,y)$ and every $\ell \in [L, n-s]$, $$d_H\left(TC(x)_{[s,s+\ell]},TC(y)_{[s,s+\ell]}\right) \geq \delta(\ell+1)$$ where $d_H$ denotes the Hamming distance. 
	\end{itemize}	 
\end{mydef}

\noindent From $TC_\ell$ together with a modified type of block error correcting code (see \cite[\S 3.1]{chs}), CHS construct a lagged tree code
\[
TCLag_{\ell}: \{0, 1\}^{\ell} \ra \left(\{0, 1\}^{c_{lag}}\right)^{\ell}
\]
with distance $1/3$ and lag $16\sqrt{\ell}$, for some constant $c_{lag}$ \cite[Claim 6.3]{chs}.\\

\noindent \textbf{Step 2:} Fix $1 \le \ell \le n$.  From $TCLag_\ell$, CHS construct a function
\[
TCLag_{\ell}^n: \{0, 1\}^n \ra \left(\{0, 1\}^{2c_{lag}}\right)^n
\]
that is online and satisfies a weakening of the lagged tree code property \cite[Claim 6.4]{chs}: for every pair of distinct inputs $x,y$ with split $s:=split(x,y) \le n - \ell/2$ and every $d \in [16\sqrt{\ell}, \ell/2]$,
\[
d_H\left(TC(x)_{[s,s+d]},TC(y)_{[s,s+d]}\right) \geq d/3.
\]

\noindent \textbf{Step 3:} From $TCLag_\ell^n$, CHS construct a lagged tree code
\[
TC^{n}: \{0, 1\}^n \ra \left(\{0, 1\}^{2c_{lag}j}\right)^n
\]
with distance $1/3$ and lag $2^{14}$, where $j = \mathcal{O}(\log\log(n))$.\\

\noindent \textbf{Step 4:} Finally, from $TC^n$, CHS construct a binary tree code
\[
TC^{n\prime}: \{0, 1\}^n \ra \left(\{0, 1\}^{\mathcal{O}(\log\log(n))}\right)^n
\]
with distance $1/3$.\\

\begin{myrmk}\label{rmk_larger_alphabet}
The above construction, as well as the decoding algorithm in Theorem \ref{thm_binary_decode}, are easily generalized to the case where we replace $TC_\Z$ with a tree code having larger coefficients, or equivalently, we replace the family of tree codes $TC_l$ with a family having slightly larger alphabets.  We construct such codes in \S \ref{section_cyclotomy}--\ref{section_transcendence} below.  Starting from a tree code
\[
TC_{\alpha, \ell}: \left(\{0, 1\}^{\ell^\alpha}\right)^{\ell} \ra \left(\{0, 1\}^{c_1\ell^{\alpha}}\right)^{\ell}
\]
in place of $TC_\ell$ for some constants $c_1, \alpha > 0$, we can replace $TCLag_\ell$ with a lagged tree code
\[
TCLag_{\alpha, \ell}: \{0, 1\}^{\ell} \ra \left(\{0, 1\}^{c_2}\right)^{\ell}
\]
having lag $16\ell^{\alpha/(\alpha+1)}$ for some constant $c_2$, using $TC_{\alpha, \ell^{1/(\alpha+1)}}$ in place of $TC_{\sqrt{\ell}}$.  We can then use this to replace $TCLag_{\ell}^n$ with an analogous type of online function
\[
TCLag_{\alpha, \ell}^n: \{0, 1\}^n \ra \left(\{0, 1\}^{2c_2}\right)^n
\]
having lag $16\ell^{\alpha/(\alpha+1)}$, and then replace $TC^n$ with a lagged tree code
\[
TC_\alpha^{n}: \{0, 1\}^n \ra \left(\{0, 1\}^{2c_2j}\right)^n
\]
having lag $64^{\alpha +1}$ for some $j = \mathcal{O}(\log\log(n))$, by replacing the sequence $\ell_1, \dots, \ell_j$ used in its construction with a sequence defined recursively by $\ell_{j} := 2n$, $\ell_{i-1} := 32\ell_i^{\alpha/(\alpha+1)}$ and $32\ell_1^{\alpha/(\alpha+1)} \le 64^{\alpha+1}$.
\end{myrmk}

\noindent \textbf{Decoding the binary tree codes:}

\begin{mydef}\label{definition_decoding}
(\textit{Decoding algorithm for a tree code})
A decoding algorithm for a tree code $T: \Sigma_1^n \ra \Sigma_2^n$ correcting a $\gamma(n)$ fraction of errors with success probability $p(n)$ is a probabilistic algorithm which inputs $\hat{x} \in \Sigma_2^q$ and $(m'_1, \dots, m'_{r}) \in \Sigma_1^r$ for some $r < q \le n$ and outputs a value $m'_{r+1}$, such that if there exists $m \in \Sigma_1^q$ satisfying $m_i = m'_i$ for all $i \in [r]$ and $d_H(T(m)_{[i, q]}, \hat{x}_{[i, q]}) < \gamma(n)(q-i+1)$ for all $i \in [r+1]$, then $m'_r = m_r$ with probability at least $p(n)$.  Here $\hat{x}$ is interpreted as an errored output of $T$, truncated to the first $q$ symbols, and $m'$ is its partial decoding.
\end{mydef}

\begin{mythm}\label{thm_binary_decode}
Let us be given a family of decoding algorithms for the $TC_{\ell}$ correcting a $\gamma(\ell)$ fraction of errors with success probability $p(\ell)$, where $\gamma(\ell)$ and $p(\ell)$ are nonincreasing.  Then for all $n$, one can define a decoding algorithm for the binary tree code $TC^{n\prime}$ correcting a $2\gamma(\sqrt{n})/3$ fraction of errors with success probability $p(n)^{O(\log^*(n))}$.  Furthermore, the runtime of this algorithm is polynomial in $n$ and the runtime of the decoding algorithm for $TC_{\ell}$ for $\ell \le n$.
\end{mythm}

\noindent In particular, using our decoding algorithm (correcting a $\gamma(\ell)=\Omega\left(\ell^{-1/2}\log^{-1/2}(\ell)\right)$ fraction of errors) for the CHS integer tree codes, we obtain a decoding algorithm for $TC^{n\prime}$ correcting a $\Omega\left(n^{-1/4}\log^{-1/2}(n)\right)$ fraction of errors in expected time polynomial in $n$. Observing that a $\Omega\left(n^{-1/4}\log^{-1/2}(n)\right)$ fraction of errors is a $\Omega\left(n^{3/4}\log^{-1/2}(n)\right)$ number of errors, we see that Theorem \ref{thm_intro_integer_decoding} implies Theorem \ref{thm_intro_binary_decoding}.

\begin{proof}[Proof sketch]
We proceed one step at a time, making reference to the precise constructions of the codes in each step (see \cite[\S 6]{chs}). \\

\noindent \textit{Step 1:} We define a decoding algorithm for a lagged tree code the same as a decoding algorithm for a non-lagged tree code, except that we add the condition $r < q - L$, where $L$ is the lag.  We obtain such an algorithm for $TCLag_{\ell}$ as follows.  Given an input interpreted as an errored truncated output of $TCLag_{\ell}$, to each complete block (supposedly) output by $ECC$, apply an algorithm to decode $ECC$ up to distance $5/12$.  Such an algorithm is easily obtained from a decoding algorithm for the algebraic geometry codes used in \cite[Appendix A]{chs}, e.g., the Shokrollahi-Wasserman algorithm \cite{SW99}.  Then concatenate the corrected blocks to form an errored truncated output of $TC_{\sqrt{\ell}}$ and decode it using the given decoding algorithm.\\

\noindent \textit{Step 2:} We define a decoding algorithm for this restricted type of lagged codes in the obvious way.  We obtain such an algorithm for $TCLag_{\ell}^n$ as follows.  Given an input interpreted as an errored truncated output of $TCLag_{\ell}^n$ and a target index $q$, let $j$ be such that the $j$-th input block $m_j$ contains index $q$.  If $j$ is even, write $j = 2i$ and apply the decoding algorithm for $TCLag_{\ell}$ to the input's errored value of $e_i$.  If $j$ is odd, instead write $j = 2i-1$ and use $o_i$.\\

\noindent \textit{Step 3:} We obtain a decoding algorithm for the lagged tree code $TC^{n}$ as follows.  Given an input interpreted as an errored truncated output of $TC^{n}$ of length $q$ and a target index $r$, let $i_0 \in [j]$ be such that $16\sqrt{\ell_{i_0}} \le q - r - 1 \le \ell_{i_0}/2$, which exists by definition of the $\ell_i$ and the lagged condition $r < q - 2^{14}$.  Apply the decoding algorithm for $TCLag_{\ell}^n$ to the coordinates $t_{i_0}$ of the input.\\

\noindent \textit{Step 4:} If $r < q - 2^{14}$, use the decoding algorithm for $TC^{n}$; otherwise use majority rule to decode the trivial tree code added by step 4.
\end{proof}

\section{The Decoding Algorithm}\label{section_alg_alg}
In this section, we devise an algebraic algorithm to recover the integer string $(z_0,z_1,\ldots, z_{n-1})$ given an erroneous version $((\hat{z}_0,\hat{a}_0),(\hat{z}_1,\hat{a}_1),\ldots,(\hat{z}_{n-1},\hat{a}_{n-1}))$ of its encoding $((z_0,a_0),(z_1,a_1),\ldots,$ $(z_{n-1},a_{n-1}))$. The number of errors  (that is, the number of pairs such that $(\hat{z}_i,\hat{a}_i) \neq (z_i,a_i)$), possibly adversarial, is assumed to be bounded by $\mathcal{O}(\sqrt{n/\log(n)})$.\\

\noindent We assume the evaluations (correct or corrupt) are bounded integers lying in $[-Z, Z]$. The Newton coefficients (correct or corrupt) are consequently bounded integers lying in $[-Z2^n, Z2^n]$.

\subsection{Evaluation Non-error Location Parameters}
\label{sec_params}
Our goal is to first identify a sufficiently large set of correct evaluation indices $I_z$. In this subsection, we outline the desired non-error location framework and determine a consistent set of parameters to eventually correct $\Omega(\sqrt{n/\log(n)})$ errors. \\ 

\noindent \textbf{Input:} $((\hat{z}_0,\hat{a}_0),(\hat{z}_1,\hat{a}_1),\ldots,(\hat{z}_{n-1},\hat{a}_{n-1}))$.\\
  
\noindent \textbf{Parameters:} The evaluation non-error location algorithm is parametrized by  $\alpha, \beta, \delta, \epsilon \in [n]$, $r \in \R$ satisfying:
\begin{enumerate}[(1)]
  \item $0 < r < \alpha\epsilon/n$
  \item $\alpha + \beta - 1 \le n/2$
  \item $\delta \le -\frac{\alpha^2\epsilon}{n} - \frac{\alpha\beta\epsilon}{n} - r(n - (\alpha + \beta) + 1) + \alpha - 1$
  \item $\alpha \le \beta$.
\end{enumerate}

\noindent \textbf{Output:} Suppose there exists $(z_0,z_1,\ldots,z_{n-1}) \in [-Z, Z]^n$ such that its encoding and the received corrupted encoding have Hamming distance at most $\epsilon$. Then with probability at least $1 - ne^{-nr^2}$, the output is a set of indices $I \subset [0, \alpha)$ such that $|I| = \delta - \epsilon + 1$ and such that if $i \in I$, then $\hat{z}_i = z_i$.  Otherwise, the output is undefined.\\

\noindent \textbf{Parameters to correct $\Omega(\sqrt{n/\log(n)})$ errors:}  We next determine a consistent set of algorithm parameters to correct $\Omega(\sqrt{n/\log(n)})$ errors. Suppose we want a success probability $\ge 1 - n^{-c}$ for some $c > 0$, and we want
\begin{equation}\tag{*}
\frac{\delta - \epsilon + 1}{\alpha} \ge 1/2.
\end{equation}
This implies $|I|/\alpha \ge 1/2$ where $I$ is the output, so that we identify at least a $1/2$ fraction of correct indices in $[0, \alpha)$. \\ 

\noindent To achieve success probability $\ge 1 - n^{-c}$, it suffices to set
\[
r = \sqrt{\frac{c+1}{n}\ln(n)}.
\]

\noindent We now wish to find values $\alpha, \beta, \delta, \epsilon \in [n]$ satisfying conditions (1)--(4) and (*).  For simplicity, we will set $\alpha = \beta$.  To satisfy (3), we set
\[
\delta = \left\lfloor -2\frac{\alpha^2\epsilon}{n} - r(n - 2\alpha + 1) + \alpha - 1 \right\rfloor.
\]

\noindent Now to satisfy (*), it suffices to ensure
\[
\frac{-2\alpha^2\epsilon/n - rn + \alpha - 2 - \epsilon + 1}{\alpha} \ge \frac{1}{2}.
\]
Solving for $\epsilon$, it suffices to satisfy
\[
\epsilon \le \frac{\alpha/2 - rn - 1}{2\alpha^2/n + 1}.
\]

\noindent We will set $\alpha = \lfloor a\sqrt{n\ln(n)} \rfloor$ for some $a \in \R$.  It suffices to ensure
\[
\epsilon \le \frac{(a/2 - \sqrt{c+1})\sqrt{n\ln(n)} - 2}{2a^2\ln(n) + 1}.
\]
Choosing $a$ so that $a/2 - \sqrt{c+1} = 1$, we find the solution
\begin{align*}
\alpha &= \beta = \left\lfloor (2 + 2\sqrt{c+1})\sqrt{n\ln(n)} \right\rfloor \\
\epsilon &= \left\lfloor \frac{\sqrt{n}}{(8c + 16\sqrt{c+1} + 17)\sqrt{\ln(n)}} - 2 \right\rfloor
\end{align*}
with $\delta$ and $r$ as above.  Note that this satisfies conditions (1) and (2) as well, at least for sufficiently large $n$ (as a function of $c$).\\

\noindent \textbf{Working modulo a large prime: } We next discuss how to find the desired indices $I_z$ with no evaluation errors. The indices are obtained one at a time, in each step calling on an algebraic algorithm. These algebraic components are easier to express in terms of field arithmetic than $\Z$-module arithmetic. To this end, we pick and work modulo a prime number $p$, choosing $p > 2^{n^2}n^{n/2}$ to ensure that the additive uncertainty lemma \ref{newton_uncertainty_lemma} still holds. 

\begin{mylemma}\label{lemma_sparsity}
	Fix $n \in \N$.  Let $p > 2^{n^2}n^{n/2}$ be prime.
	Let $\sum_{j}a_j \binom{x}{j} \in \F_p[x]$ be a nonzero polynomial of degree less than $n$ and sparsity $s \ge 1$ in the Newton basis. Let $0 \le c $ be the smallest integer such that $\sum_{j}a_j \binom{c}{j} \neq 0$.  Then $\sum_{j}a_j \binom{x}{j}$ has at most $s-1$ distinct roots in $c, c+1, \dots, n - 1 \pmod{p}$, noting $c \le n - 1$ by degree considerations.
\end{mylemma}
\begin{proof} 
	It suffices to show the following: Let $0 \le a_1 < a_2 < \cdots < a_\ell < n$ and $0 \le b_1 < b_2 < \cdots < b_\ell < n$ be integers such that $a_i \ge b_i$ for each $i \in [\ell]$. Then the $\ell \times \ell$ matrix $M$ with entries $M_{i, j} = \binom{a_i}{b_j}$ has $\det M \not\equiv 0 \pmod{p}$. From Lindstr\"om-Gessel-Viennot (Lemma \ref{gessel_viennot}), $\det M > 0$.  We claim $\det M < p$.  The entries $\binom{a_i}{b_j}$ all have absolute value bounded by $2^{n-1} \le 2^n$, so by Hadamard's inequality, $|\det M| \le 2^{n\ell}\ell^{\ell/2} \le 2^{n^2}n^{n/2} < p$. 
\end{proof}

\noindent To ensure the runtime is still polynomial in $n$, the prime $p$ should satisfy $\log(p) = n^{\mathcal{O}(1)}$. All constraints are met for $\max\{2^{n^2}n^{n/2}, 2Z2^n\} < p \le 2\cdot \max\{2^{n^2}n^{n/2}, 2Z2^n\}$, and by the the prime number theorem we may generate and test to find a prime in that interval at random, using the AKS algorithm to test.\\

\noindent In light of Lemma \ref{lemma_sparsity}, we may recast the decoding problem to be over $\F_p$. Abusing notation, in the subsequent subsection, we will use $z_i,a_i,z_i^\prime,a_i^\prime$ interchangeably to denote integers or their respective reductions modulo $p$. The meaning will be clear from context.

\subsection{Error Location by the Polynomial Method}\label{sec_error_location}

\noindent The disagreement between the evaluations and Newton basis coefficients in the received string is quantified by the vector $y=(\hat{y}_0,\hat{y}_1,\ldots,\hat{y}_{n-1})\in\F_p^n$ where $$\hat{y}_i:= \hat{z}_i  - \sum_{j=0}^{n-1} \hat{a}_j\binom{j}{i}.$$

\noindent In light of the linearity of the encoding, in place of solving the decoding problem with input $(\hat{z},\hat{a}) \in \F_p^n \times \F_p^n$, it suffices to solve the decoding problem with input $(\hat{y},0) \in \F_p^n \times \F_p^n$.\\

\noindent Under this substitution, the condition on the output translates as follows. Suppose there exists $e(x)= \sum_{j=0}^{n-1} e_j \binom{x}{j} \in \F_p[x]$ of degree less than $n$, $\epsilon$-sparse in the Newton basis, such that its evaluation and the received corrupted evaluation have Hamming distance at most $\epsilon$. Our goal is to output a set of indices $I \subset [0, \alpha)$ such that $|I| \ge \delta - \epsilon + 1$ and such that if $i \in I$, then $\hat{y}_i = e(i)$. Informally, we want to identify a set of indices $I$ at which $\hat{y}$ is correct, given that it is an erroneous version of the evaluation of some sparse polynomial in the Newton basis. \\

\noindent We begin by initializing $I = \emptyset$. We call the polynomial method below $\delta - \epsilon + 1$ times, at each iteration with the input $\hat{y}$ and the current index set $I$. The polynomial method returns one index at each iteration, which we append to $I$. The success probability $1 - ne^{-nr^2}$ comes from the union bound, noting that we make $\delta - \epsilon + 1 \le \alpha \le n$ calls.\\

\noindent \textbf{The Polynomial Method:} We next describe the algebraic algorithm to locate one index where $\hat{y}$ is correct. 
As before, fix $n \in \N$, a large prime $p > 2^{n/2}n^{n/2}$ and parameters $\alpha, \beta, \delta, \epsilon \in [n]$, $r \in \R$ satisfying:
\begin{enumerate}[(1)]
  \item $0 < r < \alpha\epsilon/n$
  \item $\alpha + \beta - 1 \le n/2$
  \item $\delta \le -\frac{\alpha^2\epsilon}{n} - \frac{\alpha\beta\epsilon}{n} - r(n - (\alpha + \beta) + 1) + \alpha - 1$
  \item $\alpha \le \beta$.
\end{enumerate}

\noindent \textbf{Input:} A vector $\hat{y} \in \F_p^n$ and a subset $H \subset [0, \alpha)$ such that $|H| \le \delta - \epsilon$.\\

\noindent\textbf{Output:} Suppose there exists a polynomial $e(x) \in \F_p[x]$ of degree less than $n$, $\epsilon$-sparse in the Newton basis, such that $d_H(\hat{y}, (e(0), e(1), \dots, e(n-1))) \le \epsilon$.  Then with probability at least $1 - e^{-nr^2}$, the output is an index $i_0 \in [0, \alpha)$ such that $i_0 \notin H$ and $\hat{y}_{i_0} = e(i_0)$.  Otherwise, the output is undefined.\\

\noindent\textbf{Algorithm:} We adapt the Shokrollahi-Wasserman algorithm \cite{SW99}.\\

\noindent Choose a set of indices $J \subset [\beta, n)$ with $|J| = n - (\alpha + \beta) + 1$ uniformly at random.  Consider the problem of finding polynomials $b(x), c(x) \in \F_p[x]$ such that
\begin{itemize}
  \item $b(x)$ is supported on $J \cup [0, \beta)$ in the Newton basis (i.e., on the basis polynomials $\binom{x}{j}$ for $j \in J \cup [0, \beta)$)
  \item $\deg c(x) \le \alpha - 1$
  \item for all $i \in H$, $c(i) = 0$
  \item for all $i \in [0, n) \setminus H$, $b(i) + c(i)\hat{y}_i = 0$.
\end{itemize}
This problem corresponds to a linear system with $n+1$ variables ($|J| + \beta = n - \alpha + 1$ from $b(x)$ and $\alpha$ from $c(x)$) and $n$ constraints.  Hence we can use linear algebra to find a solution with $b(x), c(x)$ not both 0.\\

\noindent Output the smallest value of $i \in [0, \alpha)$ for which $c(i) \neq 0$.

\subsection{Correctness}
Let $e(x)$ be a polynomial of the form described in the output condition.  Define $g(x) := b(x) + c(x)e(x)$.  Let $g|_n(x)$ denote the polynomial formed from $g(x)$ by setting all Newton basis coefficients outside of $[0, n)$ to 0, i.e., $g|_n(x) := \sum_{i=0}^{n-1} g_i \binom{x}{i}$ where the $g_i$ are defined by $g(x) = \sum_{i=0}^\infty g_i \binom{x}{i}$.  Note that $g(i) = g|_n(i)$ for all $i \in [0, n)$.\\

\noindent Let $K \subset [0, n)$ be the set of indices in $[0, n)$ at which $c(x)e(x)$ has nonzero coefficient in the Newton basis, i.e., the set of $i \in [0, n)$ such that the basis polynomial $\binom{x}{i}$ has nonzero coefficient in $c(x)e(x)$.  For any basis polynomial $\binom{x}{i}$, $c(x)\binom{x}{i}$ has sparsity at most $\deg c(x) + 1 \le \alpha$.  Indeed, it has degree at most $\deg c(x) + i$ and is zero at all $x \in \{0, \dots, i - 1\}$, so it can be written in terms of $\binom{x}{i}, \binom{x}{i+1}, \dots, \binom{x}{\deg c(x) + i}$.  Hence $|K| \le \alpha \epsilon$.\\

\noindent Now $g(x)|_n$ is supported on $J \cup [0, \beta) \cup K$ in the Newton basis.  Thus its sparsity is at most $|J \cup [0, \beta) \cup K|$.  We wish to upper bound this quantity.

\begin{mylemma}
Let $S, T \subset U$, where $S$ and $U$ are fixed and $T$ is a uniform random size $\tau$ subset of $U$.  Then for any $0 < r < |S|/|U|$,
\[
\Pr\left[|S \cup T| \ge |S| - \frac{|S|\tau}{|U|} + (r + 1)\tau\right] \le e^{-2\tau r^2}.
\]
\end{mylemma}
\begin{proof}
It suffices to prove a probabilistic lower bound on $|S \cap T|$.  We will instead prove a probabilistic upper bound on $|(U \setminus S) \cap T|$.\\

\noindent Let $t_1, \dots, t_\tau$ denote the elements of $T$, which are sampled without replacement from $U$.  We have $|(U \setminus S) \cap T| = \sum_{i=1}^\tau \mathbf{1}_{t_i \in U \setminus S}$.  We can think of the indicator variables $\mathbf{1}_{t_i \in U \setminus S}$ as a random sample of $\tau$ values without replacement from a population containing $|U \setminus S|$ ones and $|S|$ zeroes.  Then by a variant of the Chernoff bound for samples without replacement due to Hoeffding \cite[\S 6 and Theorem 1]{hoeffding}, for any $0 < r < 1 - \frac{|U \setminus S|}{|U|} = \frac{|S|}{|U|}$,
\[
\Pr\left[ \frac{|(U \setminus S) \cap T|}{\tau} - \frac{|U \setminus S|}{|U|} \ge r\right] \le e^{-2\tau r^2}.
\]

\noindent To complete the proof, apply the relation
\[
|S \cup T| = |S| + |T| - |S \cap T| = |S| + |T| - (|T| - |(U \setminus S) \cap T|) = |S| + |(U \setminus S) \cap T|.
\]
\end{proof}

\noindent Note that the conclusion of the lemma continues to hold if we replace $|S|$ with an upper bound on $|S|$.  Applying the lemma with $U = [\beta, n)$, $T = J$, $S = K \setminus [0, \beta)$, $r = r$, and the upper bound $|S| \le \alpha\epsilon$, we find
\begin{align*}
&\Pr\left[|J \cup (K \setminus [0, \beta))| \ge \alpha\epsilon - \frac{\alpha\epsilon(n - (\alpha + \beta) + 1)}{n - \beta} + (r + 1)(n - (\alpha + \beta) + 1)\right] \\ &\le e^{-2 (n - (\alpha + \beta) + 1)r^2} \le e^{-nr^2}.
\end{align*}
using conditions (1) and (2).\\

\noindent Assume the inner equality does not hold; this gives the probabilistic condition in the output description.  Simplifying, we have
\begin{align*}
|J \cup (K \setminus [0, \beta))|
&< \alpha\epsilon - \frac{\alpha\epsilon(n - (\alpha + \beta) + 1)}{n} + (r + 1)(n - (\alpha + \beta) + 1) \\
&< \frac{\alpha^2\epsilon}{n} + \frac{\alpha\beta\epsilon}{n} + r(n - (\alpha + \beta) + 1) + n - \alpha - \beta + 1,
\end{align*}
so
\begin{align*}
|J \cup [0, \beta) \cup K|
&< \frac{\alpha^2\epsilon}{n} + \frac{\alpha\beta\epsilon}{n} + r(n - (\alpha + \beta) + 1) + n - \alpha + 1 \\
&\le n - \delta
\end{align*}
by condition (3).  This gives a bound on the sparsity of $g(x)|_n$ in the Newton basis.  \\

\noindent For all $i \in [0, n) \setminus H$ such that $y_i = e(i)$, by the linear system defining $b(x)$ and $c(x)$, we have $g(i) = 0$.  Then by the condition $d_H(\hat{y}, (e(0), e(1), \dots, e(n-1))) \le \epsilon$, $g(x)|_n$ has at least $n - |H| - \epsilon$ roots in $[0, n)$.  Thus by Lemma \ref{lemma_sparsity}, $g(x)|_n$ is zero at the first
\[
\delta - \epsilon - |H| + 1
\]
values in $[0, n)$.  This number is nonzero by the input condition $|H| \le \delta - \epsilon$.\\

\noindent It follows that the Newton basis polynomials with indices in $[0, \delta - \epsilon - |H| + 1)$ have coefficient 0 in $g|_n(x)$.  We can then refine our upper bound for the sparsity of $g|_n(x)$: instead of $|J \cup [0, \beta) \cup K|$, we can upper bound it by
\[
|J \cup [\delta - \epsilon - |H| + 1, \beta) \cup K|.
\]
This bound is lower than our previous upper bound $n - \delta$ by $\min\{\delta - \epsilon - |H| + 1, \beta\}$.  So repeating the above argument, we find that $g|_n(x)$ is zero at the first
\[
\delta - \epsilon - |H| + 1 + \min\{\delta - \epsilon - |H| + 1, \beta\}
\]
values in $[0, n)$.  Iterating, we find that $g|_n(x)$ is zero at all values in $[0, \beta)$.\\

\noindent Let $i_0$ be the output value, i.e., the smallest value of $i \in \N$ for which $c(i) \neq 0$.  Because $\deg c(x) \le \alpha - 1$, we have $i_0 \in [0, \alpha)$.  By condition (4), $i_0 \in [0, \beta)$, hence $g(i_0) = g|_n(i_0) = 0$.  By the linear system defining $b(x)$ and $c(x)$, it follows that $e(i_0) = y_{i_0}$.  Finally, by the linear system defining $b(x)$ and $c(x)$, we have $i_0 \notin H$.  This proves the output's correctness.

\begin{mythm}\label{evaluation_error_theorem}
For any constant $c > 0$, there is an instantiation of the above algorithm and a function $\epsilon(n) = \Omega(\sqrt{n/\log(n)})$, specified in \S \ref{sec_params} above, with the following guarantee: Given an input $((\hat{z}_0,\hat{a}_0),(\hat{z}_1,\hat{a}_1),\ldots,(\hat{z}_{n-1},\hat{a}_{n-1}))$, where all $\hat{z}_i \in [-Z, Z]$ and $\hat{a}_i \in [-Z2^n, Z2^n]$, if there exists $z = (z_0, \dots, z_{n-1}) \in [-Z, Z]^n$ such that the encoding of $z$ and the input have Hamming distance at most $\epsilon(n)$, then with probability at least $1 - n^{-c}$, the output is a set of indices $I \subset [0, \alpha(n))$ such that $|I|/\alpha(n) \ge 1/2$ and such that if $i \in I$, then $\hat{z}_i = z_i$.  Here $\alpha(n)$ is some positive function of $n$.  Furthermore, the runtime of the algorithm is polynomial in $n$ and $\log(Z)$.
\end{mythm}

\subsection{Newton Non-error Location by Reversal}
Our next goal is to identify a sufficiently large set of correct Newton basis indices $I_a$.  The inversion formula $$z_i  = \sum_{j=0}^{n-1}a_j\binom{i}{j},\qquad a_j = \sum_{i=0}^{n-1} (-1)^{j-i}\binom{j}{i}z_i$$
suggests that perhaps we could use a similar approach to locating correct evaluation indices. In fact, we show that the problem of locating correct Newton basis indices can be translated so that we can apply the above algorithm unchanged. The key is that the matrix of the transformation from the evaluation to Newton bases is a nonzero row and column scaling of the transformation from the Newton to evaluation bases. While the translation is simple in light of the inversion theorem, we develop a more general duality which allows translation for a wide class of codes, including the number theoretic variants discussed in \S \ref{section_cyclotomy}--\ref{section_transcendence}. Furthermore, additive uncertainty is invariant under this scaling, meaning it could be a tool in proving the additive uncertainty of other transformations. \\

\noindent \textbf{Uncertainity Morphisms:}
Define an uncertainty problem to be a linear transformation $A: F^n \ra F^n$ over some field $F$.  We are interested in ``decoding'' pairs of vectors $(\hat{x}, \widehat{Ax})$, which are taken to be errored versions of some pair $(x, Ax)$.  For example, when $A$ is the evaluation map in the Newton basis over $\F_p$ at the points $[0, n)$, decoding pairs corresponds to decoding CHS tree codes. Define a morphism of uncertainty problems $(T, U): A \ra B$ to be a pair of linear transformations $T, U: F^n \ra F^n$ such that
\begin{itemize}
	\item $T$ and $U$ are given by invertible diagonal matrices and
	\item we have the commutative diagram
	\[
	\begindc{\commdiag}[500]
	\obj(0, 0)[00]{$F^n$}
	\obj(0, 1)[01]{$F^n$}
	\obj(1, 0)[10]{$F^n.$}
	\obj(1, 1)[11]{$F^n$}
	\mor{01}{11}{$T$}
	\mor{00}{10}{$U$}
	\mor{01}{00}{$A$}[\atright, \solidarrow]
	\mor{11}{10}{$B$}[\atleft, \solidarrow]
	\enddc
	\]
\end{itemize}
By applying the components $T$ and $U$ of such a morphism to a corrupted transform vector pair $(\hat{x}, \widehat{Ax})$ with respect to $A$, we get a corrupted transform vector pair $(T\hat{x}, U\widehat{Ax}) = (\widehat{Tx}, \widehat{BTx})$ with respect $B$ having the same error locations as the original pair.  This is true for inputs with no errors by the commutative diagram, and adding one error to an input pair $(\hat{x}, \widehat{Ax})$ adds an error to the same location in the output pair $(T\hat{x}, U\widehat{Ax})$ because $T$ and $U$ are diagonal.  \\

\noindent Now let $E: \F_p^n \ra \F_p^n$ denote the evaluation map in the Newton basis at the points $[0, n)$.  Then to ``reverse'' the algorithm from Theorem \ref{evaluation_error_theorem}, it suffices to define a morphism from $E$ to $E^{-1}$. That is, it suffices to find two invertible diagonal matrices $T_E,U_E$ such that the following diagram commutes:
\[
\begindc{\commdiag}[500]
\obj(0, 0)[00]{$\F_p^n$}
\obj(0, 1)[01]{$\F_p^n$}
\obj(1, 0)[10]{$\F_p^n.$}
\obj(1, 1)[11]{$\F_p^n$}
\mor{01}{11}{$T_E$}
\mor{00}{10}{$U_E$}
\mor{01}{00}{$E$}[\atright, \solidarrow]
\mor{11}{10}{$E^{-1}$}[\atleft, \solidarrow]
\enddc
\]
From the inversion formula, we can take both $T_E$ and to be the diagonal matrix with the diagonal entries alternating between $1$ and $-1$.\\

\noindent We then get the following algorithm to locate correct indices in the $\hat{a}_i$ instead of the $\hat{z}_i$: swap the vectors $\hat{a}$ and $\hat{z}$, apply the matrices $T_E$ and $U_E$ to them (i.e., negate every even coordinate), and then apply the algorithm from Theorem \ref{evaluation_error_theorem}.  Thus:\\

\begin{mythm}\label{newton_error_theorem}
For any constant $c > 0$, there is a function $\epsilon(n) = \Omega(\sqrt{n/\log(n)})$, identical to that in Theorem \ref{evaluation_error_theorem}, such that the algorithm in the above paragraph has the following guarantee: Given an input $((\hat{z}_0,\hat{a}_0),(\hat{z}_1,\hat{a}_1),\ldots,(\hat{z}_{n-1},\hat{a}_{n-1}))$, where all $\hat{z}_i \in [-Z, Z]$ and $\hat{a}_i \in [-Z2^n, Z2^n]$, if there exists $z = (z_0, \dots, z_{n-1}) \in [-Z, Z]^n$ such that the encoding of $z$ and the input have Hamming distance at most $\epsilon(n)$, then with probability at least $1 - n^{-c}$, the output is a set of indices $I \subset [0, \alpha(n))$ such that $|I|/\alpha(n) \ge 1/2$ and such that if $i \in I$, then $\hat{a}_i = a_i$.  Here $\alpha(n)$ is identical to that in Theorem \ref{evaluation_error_theorem}.  Furthermore, the runtime of the algorithm is polynomial in $n$ and $\log(Z)$.
\end{mythm}

\begin{myrmk}\label{uncertainty_morphism_remark}
	A further useful consequence of a morphism $(T,U):A\rightarrow B$ is that additive uncertainty holds with respect to $A$ if and only if it holds with respect to $B$. Often it is easier to prove uncertainty for one of $A$ or $B$; we may then infer uncertainty for the other using the morphism.
\end{myrmk}

\subsection{From Error Location to Error Correction}
Using the algorithms in the previous two sections on an input $((\hat{z}_0,\hat{a}_0),(\hat{z}_1,\hat{a}_1),\ldots,(\hat{z}_{n-1},\hat{a}_{n-1}))$ which is an erroneous version of an encoding $((z_0,a_0),(z_1,a_1),\ldots,(z_{n-1},a_{n-1}))$ having at most $\epsilon(n)$ errors, we get an integer $\alpha(n)$ and subsets $I_z, I_a \subset [0, \alpha(n))$ such that $|I_z|, |I_a| \ge \alpha(n)/2$ and the input equals its corrected encoding at the $z$ indices in $I_z$ and the $a$ indices in $I_a$.  From this, we wish to derive $z_0$, the correct first coordinate of the input to the tree code. \\

\noindent If $0 \in I_z$ or $0 \in I_a$, we merely output $\hat{z}_0$ or $\hat{a}_0$, respectively.  Otherwise, as in \S \ref{sec_error_location}, it suffices to solve this problem with $\hat{z}$ and $\hat{a}$ replaced by $\hat{y}$ and $(0, 0, \dots, 0)$, where
\[
\hat{y}_i:= \hat{z}_i  - \sum_{j=0}^{n-1} \hat{a}_j\binom{j}{i}.
\]

\noindent Note that we are now faced with a polynomial interpolation problem: we want to find a polynomial $f(x)$ supported on indices $[0, n) \setminus I_a$ in the Newton basis such that $f(i) = \hat{y}_i$ for all $i \in I_z$.  To perform the interpolation, we will use the Lindstr\"om-Gessel-Viennot Lemma, as stated in Lemma \ref{gessel_viennot}.  Our algorithm is based on the proof of Lemma 1.4 in the CHS paper \cite{chs} (reproduced above as Lemma \ref{newton_uncertainty_lemma}). \\

\noindent For $i \in [0, n)$, let $t(i) = |[0, i) \cap I_z| - |[0, i) \setminus I_a|$.  For all $i$, $t(i+1) - t(i) \in \{-1, 0, 1\}$, i.e., $t(i)$ moves by unit steps.  We have $t(0) = -1$ and $t(\alpha(n)) = |I_z| - |[0, \alpha(n)) \setminus I_a| \ge \alpha(n)/2 - \alpha(n)/2 = 0$.  Thus there is a smallest index $i_0 \in [0, \alpha(n))$ such that $t(i_0) = 0$.  Then letting $0 < r_1 < r_2 < \cdots < r_d$ be the elements of $[0, i_0) \cap I_z$ and $0 = c_1 < c_2 < \cdots < c_d$ be the $d$ smallest elements of $[0, n) \setminus I_a$, we have $r_k \ge c_k$ for all $k$: otherwise, letting $r_{k_0} < c_{k_0}$ be the first counterexample, we would have $t(r_{k_0}) > 0$, contradicting the minimality of $i_0$. \\

\noindent Now the Lindstr\"om-Gessel-Viennot Lemma, as stated in Lemma \ref{gessel_viennot}, implies that there is a unique polynomial $g(x)$ supported on indices $\{c_1, c_2, \dots, c_d\}$ in the Newton basis such that $g(r_k) = \hat{y}_{r_k}$ for all $k \in [d]$.  We compute $g(x)$ by inverting the matrix $\{ \binom{r_k}{c_\ell} \mid k, \ell \in [d]\}$ and multiplying it by $(\hat{y}_{r_1}, \hat{y}_{r_2}, \dots, \hat{y}_{r_d})$.  Since $t(i_0) = 0$, we have $[0, i_0) \setminus I_a = \{c_1, \dots, c_d\}$, so $g(x)$ is in fact the unique polynomial supported on $[0, i_0) \setminus I_a$ such that $g(r_k) = \hat{y}_{r_k}$ for all $k \in [d]$.  The polynomial $\sum_{j=0}^{i_0 - 1} a_j \binom{x}{j}$ corresponding to the true Newton basis error vector $a$, restricted to $[0, i_0)$, also has this property, assuming $I_a$ and $I_z$ are correct.  Thus $g(0) = q(0) = a_0 = z_0$, so we finish the algorithm by outputting $g(0)$. \\

\begin{mythm}\label{thm_final}
For any constant $c > 0$, there is a function $\epsilon(n) = \Omega(\sqrt{n/\log(n)})$, identical to that in Theorem \ref{evaluation_error_theorem}, such that the above algorithm has the following guarantee: Given an input $((\hat{z}_0,\hat{a}_0),(\hat{z}_1,\hat{a}_1),$ $\ldots,(\hat{z}_{n-1},\hat{a}_{n-1}))$, where all $\hat{z}_i \in [-Z, Z]$ and $\hat{a}_i \in [-Z2^n, Z2^n]$, if there exists $z = (z_0, \dots, z_{n-1}) \in [-Z, Z]^n$ such that the encoding of $z$ and the input have Hamming distance at most $\epsilon(n)$, then with probability at least $1 - 2n^{-c}$, the output is $z_0$.  Furthermore, the runtime of the algorithm is polynomial in $n$ and $\log(Z)$.
\end{mythm}

After using the algorithm to find a supposed value $z_0'$ for $z_0$, we can repeat the algorithm to find a supposed value $z_1'$ for $z_1$: subtract the encoding of $(z_0', 0, 0, \dots, 0)$ from the input and remove the first coordinate before applying the algorithm again.  Note that this time, we must adjust the algorithm so that the set of evaluation points starts at 1 instead of 0, a trivial change.  Similarly, we can find supposed values $z_2', z_3', \dots, z_{n-1}'$ for $z_2, z_3, \dots, z_{n-1}$.  Note however that $z_i'$ will only be correct if there exists $z = (z_0, \dots, z_{n-1}) \in [-Z, Z]^n$ such that for all $j \in [0, i]$, $$d_H\left(TC_Z(z)_{[j, n)}, ((\hat{z}_j, \hat{a}_j), \dots, (\hat{z}_{n-1}, \hat{a}_{n-1}))\right) \le \epsilon(n-j).$$  This is a stronger condition than the condition for $z_0'$ to be correct, and indeed, it is possible for an input to decode correctly for earlier coordinates but not for later coordinates, if errors are concentrated towards the end of the input.\\

\noindent Finally, we make some observations regarding the success probability.  We can repeat the algorithm in Theorem \ref{thm_final} and take majority rule, performing two-sided amplification, to exponentially increase the probability that we get the correct value for $z_0$, with the remaining probability including the case that the algorithm fails or outputs an incorrect answer.  Alternatively, we can use the technique in the above paragraph to find a tentative correct string $(z_0', z_1', \dots, z_{n-1}')$.  If the number of errors in the input is at most $\epsilon(n)/2$, then all of $z_0', z_1', \dots, z_{n/2 -1}'$ are correct with probability at least $1 - 2n^{1-c}$, and we can check whether they are correct by verifying $d_H\left(TC_\Z(z'), ((\hat{z}_0, \hat{a}_0), (\hat{z}_1, \hat{a}_1), \dots, (\hat{z}_{n-1}, \hat{a}_{n-1})\right) \le n/2$.  This lets us perform one-sided amplification to exponentially increase the probability that we get the correct value for $z_0$, with no chance of outputting an incorrect answer. In summary, we have proven Theorem \ref{thm_intro_integer_decoding}.

\section{Decoding by Convex Optimization}\label{convex_decoding_section}
Compressed sensing had its roots in engineering practice and emerged to prominence with the breakthrough results of Candes-Tao \cite{CT06} and Donoho \cite{Don06}. The formulation convenient to our context is the ``stable recovery" result of  Candes-Romberg-Tao \cite{CRT06}, which we now briefly summarize. \\

\noindent Consider the problem of recovering an $S$-sparse $N$ dimensional vector $x\in\C^N$ from a vector of $M$ observations $y=Ax+e$ where $A$ is an $M$ by $N$ matrix with $M<<N$ and $e \in \C^M$ is an error vector with bounded $\ell_2$ norm $||e||_{\ell_2} \leq \epsilon$. One may look for $x$ using the combinatorial optimization $$\arg \min_x ||x||_{\ell_0} \ such \ that\ y=Ax+e.$$
But instead, solve for the convex optimization problem $$\arg \min_x ||x||_{\ell_1} \ such \ that\ y=Ax+e,$$ with the $\ell_0$ norm replaced with the $\ell_1$ norm. Candes and Tao \cite{CT06} formulated a sufficient condition on $A$, called the Restricted Isometry Property (RIP), that ensures that the solution $\hat{x}$ obtained by the $\ell_1$ optimization is close to the sought $x$, $$||\hat{x}-x||_{\ell_2}\leq c\epsilon$$ for some constant $c$. When there is no error, the reconstruction is exact. This RIP for $A$ is easy to state. Say the columns of $A$ are normalized to have unit $\ell_2$ norm. Then $A$ is said to satisfy RIP if there exists a constant $0\leq \delta <1$ such that for all $S$ cardinality subsets $T$ of the columns of $A$, the restriction $A_T$ of $A$ to the columns $T$ is a near isometry: $$ 1-\delta \leq \frac{||A_Tw||^2_{\ell_2}}{||w||^2_{\ell_2}}\leq 1+\delta,\ \ \ \mbox{for all } w.$$  
When $A$ is drawn from certain random ensembles, it satisfies RIP with high probability when $M = \Omega(S \log^d N)$ for some small constant $d$. For instance, $d=0$ suffices for Gaussian random matrices, $d=2$ for matrices with independent $\pm 1$ Bernoulli entries and $d=4$ when the rows are randomly drawn from a discrete Fourier matrix \cite{RV08}. \\

\noindent We next cast our decoding problem as a convex optimization problem. Call the message string $z:=(z_0,z_1,\ldots, z_{n-1})$ and its Newton basis coefficients $a:=(a_0,a_1,\ldots, a_{n-1})$. The decoder is given an erroneous version $((\hat{z}_0,\hat{a}_0),(\hat{z}_1,\hat{a}_1),\ldots,(\hat{z}_{n-1},\hat{a}_{n-1}))$ of the encoding and has to recover $z$ (or equivalently $a$).\\ 

\noindent  Write the corrupted Newton coefficients as $\hat{a}=a+v$, where $v$ is the (sparse) Newton basis error vector, and write the corrupted evaluations as $\hat{z} = z + u$, where $u$ is the (sparse) evaluation basis error vector. In light of the linearity of the encoding, we have
\[
\sum_{j=0}^{n-1} v_j\binom{i}{j} - u_i = \sum_{j=0}^{n-1} \hat{a}_j \binom{i}{j} - \hat{z}_i , \mbox{ for all } i \in  [0,n).
\]
Thus our goal is to interpolate a sparse polynomial in the Newton basis, given an erroneous version of its evaluations. Henceforth, we will use $z,a,\hat{z},\hat{a},u,v$ to denote the corresponding vectors written columnwise. Let $B$ denote the lower triangular matrix $\{\binom{i}{j} \mid i,j\in [0,n)\}$ with binomial coefficients. The decoding task can be phrased as the combinatorial optimization problem 
$$\arg\min_{u,v}\left(\left|\left|u\right|\right|_{\ell_0} + \left|\left|v\right|\right|_{\ell_0}\right), \mbox{such that}\ \hat{z}-u=B(\hat{a}-v).$$ We suggest relaxing and solving instead the convex optimization $$\arg\min_{u,v}\left(\left|\left|u\right|\right|_{\ell_1} + \left|\left|v\right|\right|_{\ell_1}\right), \mbox{such that}\ \hat{z}-u=B(\hat{a}-v).$$ 
The likeness to compressed sensing is apparent when the convex optimization is rewritten with block vectors/matrices as  
$$\arg\min_{\qbinom{u}{v}}\left|\left|\qbinom{u}{v}\right|\right|_{\ell_1}\ \mbox{such that}\ [I\ |-B]\qbinom{u}{v}=[I\ |-B]\qbinom{\hat{z}}{\hat{a}},$$ where $I$ is the identity matrix.\\ 

\noindent If $[I\ |-B]$ were to satisfy RIP, then that would allow us to correct $n/\log^d n$ errors, for some small exponent $d$. In terms of uncertainty, we could see this as an algorithmic nearly additive uncertainty principle; that is for inputs with nonzero $z_0$, 
$$Sparsity(z_0,z_1,\ldots,z_{n-1}) + Sparsity(a_0,a_1,\ldots,a_{n-1}) \geq \frac{n}{(\log n)^{O(1)}}.$$
The resulting uncertainty is weaker in that there is a small polylogarithmic loss on the right. However, it is algorithmic: the convex optimization corrects $\frac{n}{(\log n)^{O(1)}}$ errors, nearly that guaranteed by the constant distance property. Detailed accounts on the relation between error correction using convex optimization and uncertainty are in \cite{DL92,CRT062}. \\

\noindent The $\ell_1$ optimization algorithm can be adapted to erasure decoding in a straightforward manner: restrict the rows of $[I\ |-B]$ to the rows with no error. Likewise, in the general error correction task, we may use randomization to pick the a subset of rows of $[I\ |-B]$; knowing only a few rows scaling roughly with the number of errors suffices.\\

\noindent \textbf{Online RIP:} It is however not clear if $[I\ | -B]$ satisfies RIP. We cannot look to ensembles known to satisfy RIP for they are not of the form $[I\ | -A]$ where $A$ is lower triangular to enforce the online property. This raises the question as to if there are random ensembles or explicit constructions of lower triangular matrices $A$ such that $[I\ | -A]$ satisfies RIP. \\

\noindent We present candidate explicit triangular matrices using number theoretic methods.  In \S~\ref{cyclotomic_subsection}, the matrices are lower triangular matrices arising in the LU decomposition of Fourier matrices. In \S~\ref{sunflower_subsection}, they are certain transcendental matrices inspired by the arrangement of leaves on plant stems \cite{Mit77}. They arise in the LU decomposition of Fourier like matrices, but with carefully chosen transcendental numbers on the unit circle replacing roots of unity. These are extended in \S~\ref{unit_circle_subsection}, where the transcendental numbers on the unit circle are generated by a process that is pseudorandom (assuming the ABC conjecture). Further theoretical and experimental investigations are warranted to determine the viability of the  convex optimization decoding framework. \\

\noindent We also remark that one may need to formulate an alternate definition for online RIP, besides just RIP restricted to online matrices.  This is motivated by the fact that Baraniuk et~al.'s proof of RIP for Gaussian random matrices \cite{baraniuk_et_al} does not work for their natural online analog, namely, Gaussian random matrices having shape $[C\ |\ D]$ with $C$ and $D$ lower triangular and columns having expected unit norm.  Essentially, the columns that are mostly zero contribute too much variance.  It may be possible to modify Candes and Tao's proof that RIP implies $\ell_1$-decoding \cite{CT06} to account for these mostly zero columns separately, but it is not immediately obvious how to do so. \\

\noindent \textbf{Alphabet size:} Convex optimization decoding has a further interesting feature in addition to correcting a large number of errors. We only need to work with a precision that allows for the convex optimization. It is an interesting question as to if RIP implies a bound on the  precision sufficient for convex optimization. Consequently, the alphabet size could be reduced.

\section{Tree Codes from Cyclotomy}\label{section_cyclotomy}
The Lindstr\"om-Gessel-Viennot Lemma applied to path lattices plays an important role in Cohen-Haeupler-Schulman codes. We extend the Cohen-Haeupler-Schulman framework by looking to path lattices with weights drawn from carefully chosen algebraic numbers. Wrapped in the coding theory framework of \S~\ref{alphabet_reduction_section}, we obtain binary tree codes with constant distance and polylogarithmic alphabet size.

\subsection{Gaussian q-binomials}
\noindent For nonnegative integers $r,s$, recall the Gaussian or $q$-binomial coefficients
\[ \qbinom{r}{s}_q :=  \begin{cases} 
\frac{(q^r-1)(q^{r-1}-1)\ldots(q^{r-s+1}-1)}{(q^s-1)(q^{s-1}-1)\ldots(q-1)} & , r \geq s \\
\ \ \ \ 0 & , r < s 
\end{cases}
\]  
and consider the evaluation map $$z_i  = \sum_{j=0}^{n-1}b_j\qbinom{i}{j}_{q} , i \in [0,n)$$  
with Carlitz's \cite{car73}(see also \cite{chu95}[eqn 1.3]) inversion formula $$b_j  = \sum_{i=0}^{n-1}z_j(-1)^{j-i}q^{\binom{j-i}{2}}\qbinom{j}{i}_{q} , j \in [0,n).$$ 
Consider the tree code $$TC_q:\Z^n \longrightarrow (\Z\times \Z[q])^n \ \ \ \ \ \ \ \ \ \ \ \ \ \ \ \ \ \ $$
$$\ \ (z_0,z_1,\ldots,z_{n-1})\longmapsto ((z_0,b_0),(z_1,b_1),\ldots,(z_{n-1},b_{n-1})).$$
For $TC_q$ to have distance $1/2$ is equivalent to the evaluation map satisfying additive uncertainty, that is, for all $(z_0,z_1,\ldots,z_{n-1})$ with first nonzero value at index $s$, $$Sparsity(z_s,z_{s+1},\ldots,z_{n-1}) + Sparsity(b_s,b_{s+1},\ldots,b_{n-1}) \geq n-s+1.$$ This in turn is equivalent for the determinant of the square matrix $\{\qbinom{r_i}{c_j}_q \mid i, j \in [m] \}$ to not vanish for nonnegative integers $r_1<r_2,\ldots<r_m$ and $c_1<c_2<\ldots<c_m$ with $r_i \geq c_i$ for all $i \in [m]$. We look to path graphs and the Lindstr\"om-Gessel-Viennot Lemma to better understand this determinant. We first realize the underlying matrix as a path matrix of the weighted directed graph below with distinguished (not necessarily disjoint) sets of vertices $r_1,r_2,\ldots,r_m$ and $c_1,c_2\ldots,c_m$. The $r_1^{th}$ vertex on the first column is labelled $r_1$, the $r_2^{th}$ vertex on the first column is labelled $r_2$ and so on. Likewise, the $c_1^{th}$ vertex on the diagonal is labelled $c_1$, the $c_2^{th}$ vertex on the diagonal is labelled $r_2$ and so on.\\

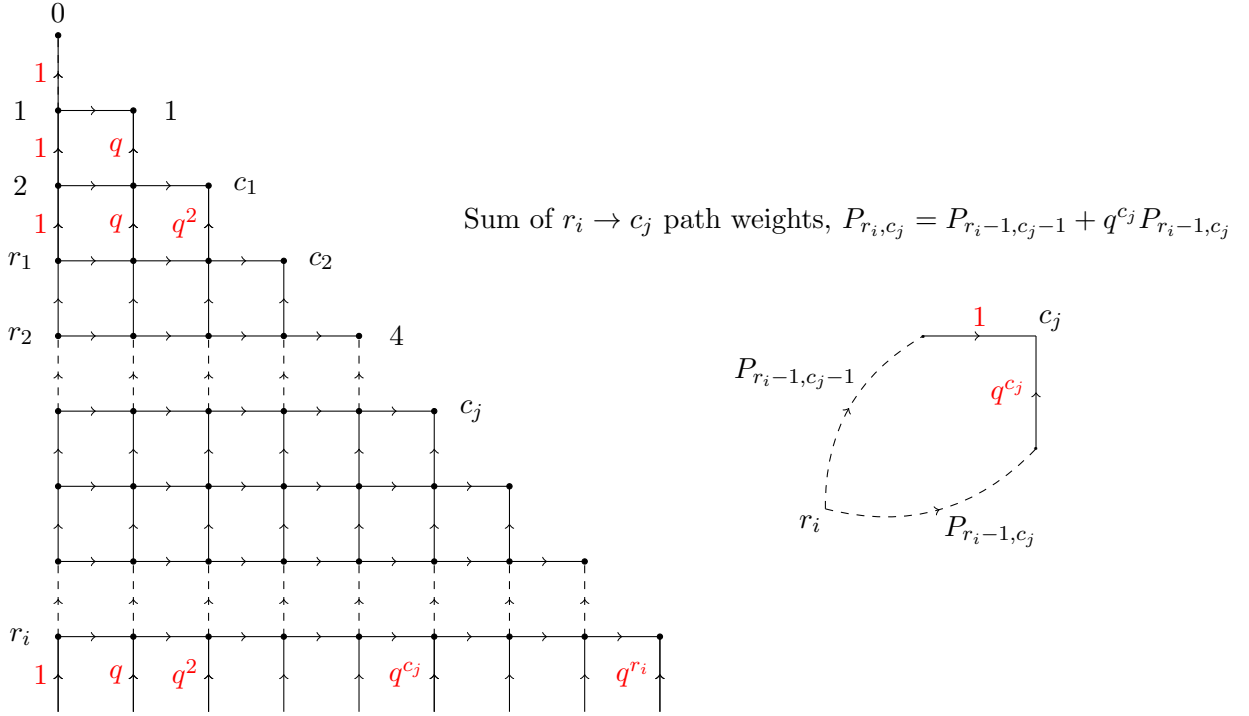
\begin{tikzpicture}
\draw[fill] (0,8) circle (1pt);
\draw[dashed,->-](0,7)--(0,8);
\draw[->-] (0,7) -- (0,8) node[midway,sloped,left,rotate=270] {\textcolor{red}{$1$}};
\node at (0, 8.3)  {$0$};

\foreach \y in {1,2,3,4,5,6,7,8}
\foreach \x in {0,...,\y}
\draw[fill] (\x,8-\y) circle (1pt) coordinate (m,\x,\y);

\foreach \y in {1,2,3}
\foreach \x in {0,...,\y}
\draw[->-](\x,7-\y) -- (\x,8-\y);

\foreach \x in {0,1,2,3,4}
\draw[dashed,->-](\x,3)--(\x,4);

\foreach \y in {5,6}
\foreach \x in {0,...,\y}
\draw[->-](\x,7-\y) -- (\x,8-\y);

\foreach \x in {0,1,2,3,4,5,6,7}
\draw[dashed,->-](\x,0)--(\x,1);

\foreach \y in {8}
\foreach \x in {0,...,\y}
\draw[->-](\x,7-\y) -- (\x,8-\y);

\foreach \y in {1,2,3,4,5,6,7,8}
\foreach \x in {1,...,\y}
\draw[->-](\x-1,8-\y) -- (\x,8-\y);

\draw[->-] (0,-1) -- (0,0) node[midway,sloped,left,rotate=270] {\textcolor{red}{$1$}};

\draw[->-] (1,-1) -- (1,0) node[midway,sloped,left,rotate=270] {\textcolor{red}{$q$}};

\draw[->-] (2,-1) -- (2,0) node[midway,sloped,left,rotate=270] {\textcolor{red}{$q^2$}};

\draw[->-] (0,6) -- (0,7) node[midway,sloped,left,rotate=270] {\textcolor{red}{$1$}};

\draw[->-] (1,6) -- (1,7) node[midway,sloped,left,rotate=270] {\textcolor{red}{$q$}};

\draw[->-] (0,5) -- (0,6) node[midway,sloped,left,rotate=270] {\textcolor{red}{$1$}};

\draw[->-] (1,5) -- (1,6) node[midway,sloped,left,rotate=270] {\textcolor{red}{$q$}};

\draw[->-] (2,5) -- (2,6) node[midway,sloped,left,rotate=270] {\textcolor{red}{$q^2$}};

\draw[->-] (5,-1) -- (5,0) node[midway,sloped,left,rotate=270] {\textcolor{red}{$q^{c_j}$}};

\draw[->-] (8,-1) -- (8,0) node[midway,sloped,left,rotate=270] {\textcolor{red}{$q^{r_i}$}};

\node at (-0.5, 7)  {$1$}; 
\node at (-0.5, 6)  {$2$}; 
\node at (-0.5, 5)  {$r_1$}; 
\node at (-0.5, 4)  {$r_2$}; 
\node at (-0.5, 0)  {$r_i$}; 

\node at (1.5,7) {$1$}; 
\node at (2.5,6) {$c_1$}; 
\node at (3.5,5) {$c_2$}; 
\node at (4.5,4) {$4$}; 
\node at (5.5,3) {$c_j$};

\node at (10,1.5) {$r_i$};
\node at (11.5,4) {$.$};
\node at (13.2,4.2) {$c_j$};
\node at (13,2.5) {$.$};
\draw[->-](11.5,4) -- (13,4) node[midway,above]{\textcolor{red}{$1$}};
\draw[->-](13,2.5) -- (13,4) node[midway,sloped,left,rotate=270]{\textcolor{red}{$q^{c_j}$}};
\draw[dashed,->-](10.2,1.7) to[bend left] (11.5,4);
\draw[dashed,->-](10.2,1.7) to[bend right] (13,2.5);

\node at (9.8,3.5) {$P_{r_i-1,c_j-1}$};

\node at (12.4,1.4) {$P_{r_i-1,c_j}$};

\node at (10.5,5.5) {Sum of $r_i \ra c_j$ path weights, $P_{r_i,c_j} = P_{r_i-1,c_j-1}+q^{c_j}P_{r_i-1,c_j}$};

\end{tikzpicture}\\

\noindent The horizontal edges are directed from left to right and have weight $1$. The vertical edges are directed from bottom to top. The left most column of edges have weight $1$, the second left most have weight $q$, the third left most have weight $q^2$ and so on. The path matrix $M$ has entries $M_{i,j}$ equal to the sum $P_{i,j}$ of the weights of paths from $r_i$ to $c_j$. From the picture on the right, it satisfies the recurrence relation $P_{r_i,c_j} = P_{r_i-1,c_j-1}+q^{c_j}P_{r_i-1,c_j}$, which is identical to the recurrence $$\qbinom{r_i}{c_j}_q = \qbinom{r_i-1}{c_j-1}_q +q^{c_j}\qbinom{r_i-1}{c_j}_q.$$ The boundary conditions agree too; $P_{r_i,0}=\qbinom{r_1}{0}_q=1$ and if $r_i=c_i$, $P_{r_i,c_i} = \qbinom{r_i}{c_i}_q=1$. Hence $M_{i,j} = \qbinom{r_i}{c_j}_q$. Curiously, it is now evident from the path matrix formulation that the $q$-binomials are indeed polynomials in $q$, not merely rational functions as suggested by the definition. By the Lindstr\"om-Gessel-Viennot Lemma, the determinant of the path matrix is $$\det(M) = \sum_{\substack{vertex\ disjoint\\ path\ systems\ \mathcal{P}}}  \prod_{P_i\in \mathcal{P}} w(P)$$
where the geometry enforces every vertex disjoint path system $\mathcal{P}$ to consist of $m$ vertex disjoint paths $\{P_i \mid i \in [m]\}$ with $P_i$ connecting $r_i$ to $c_i$. Consequently $\det(M)$ is a polynomial in $q$ with nonnegative integer coefficients. Furthermore, it is a monic polynomial with the highest degree monomial contributed by the path system where each $P_i$ starting from $r_i$ goes right for $c_i$ steps and then turns up. The degree of this term is $\sum_{i=1}^{m}c_i(r_i-c_i)$. In summary, we may claim Lemma \ref{q-binomial_lemma}, where the bound on the sum of coefficients follows by the $q$-binomial degenerating to the binomial at $q=1$.

\begin{mylemma}\label{q-binomial_lemma}
	For nonnegative integers $r_1<r_2<\ldots<r_m$ and $c_1<c_2<\ldots<c_m$ with $r_i \geq c_i$ for all $i \in [m]$, 	$$\det\left(\left\{\qbinom{r_i}{c_j}_q \mid i, j \in [m] \right\}\right) = q^{d}+ \sum_{k<d}w_kq^k \in\Z[q]$$ where $d:=\sum_{i=1}^{m}c_i(r_i-c_i)$, $w_k \geq 0$ and $\sum_{k<d}w_k = det\left(\left\{\binom{r_i}{c_j} \mid i, j \in [m] \right\}\right)-1$. 
\end{mylemma}

\subsection{Tree codes from cyclotomic units}\label{cyclotomic_subsection}
For the first number theoretic construction of tree codes of depth $n$, choose a prime number $\ell > n^3$ and substitute $q=\zeta_\ell \in \C$ for a primitive $\ell^{th}$ root of unity and consider $TC_{\zeta_\ell}$. Such a prime $\ell$ in the interval $n^3 <\ell \leq 2n^3$ can be found deterministically by exhaustively searching using Bertrand's postulate to gurantee existence and the AKS algorithm \cite{AKS04} to test for primality. The requirement that $\ell$ be prime is for ease of exposition. If $\ell$ were not prime, it suffices to assume that the Euler totient $\phi(\ell) \geq n^3$. 

\begin{mythm}\label{tree_codes_cyclotomic_theorem}
	Fix a positive integer $n$ and a prime $\ell>n^3$. The tree code $TC_{\zeta_\ell}:\Z^n \longrightarrow (\Z\times \Z[\zeta_\ell])^n$ has distance $1/2$.
\end{mythm}
\begin{proof}
	Consider nonnegative integers $r_1<r_2<\ldots<r_m$ and $c_1<c_2<\ldots<c_m$ with $c_i \le r_i \le n$, for all $i \in [m]$. The minimal polynomial of $\zeta_\ell$ over $\Q$ has degree $\ell-1 \geq n^3$.  By Lemma \ref{q-binomial_lemma}, the determinant $$det\left(\left\{\qbinom{r_i}{c_j}_{\zeta_\ell} \mid i, j \in [m] \right\}\right)$$ is a monic integer polynomial in $\zeta_\ell$ of degree less than $n^3$ and is thus nonzero.
\end{proof}

\noindent\textbf{Alphabet size bounds:} The alphabet of the binary tree codes resulting from wrapping $TC_{\zeta_\ell}$ in the coding theory machinery will depend on how the size of the coefficients $b_0,b_1,\ldots,b_{n-1} \in\Z[\zeta_\ell]$ scale. We must choose a concrete representation of $\Z[\zeta_\ell]$ before discussing the size of the coefficients. Incidentally $\Z[\zeta_\ell]$ is the entire ring of integers $\mathcal{O}_\ell$ of the cyclotomic field $\Q(\zeta_\ell)$. One natural choice is the integer basis $\{1,\zeta_\ell,\ldots,\zeta_\ell^{\ell-1}\}$. In this case, it is not hard to show that the bit size of the coefficient $b_j, j\in [n]$ is bounded by a small polynomial in $n$. Hence this results in binary tree codes with polylogarithmic alphabet size, as was the case with CHS codes.

\noindent \textbf{Cyclotomic Units:} One curious observation about the tree codes $TC_{\zeta_\ell}$ is that the transformation matrices $\left\{\qbinom{i}{j}_{\zeta_\ell} \mid i,j\right\}$ and $\left\{(-1)^{j-i}\zeta_\ell^{\binom{j-i}{2}}\qbinom{i}{j}_{\zeta_\ell} \mid i,j\right\}$ have nonzero entries that are cyclotomic units. For $k \in [0,\ell)$, the elements 
$$ \epsilon_k:=\frac{1-\zeta_\ell^k}{1-\zeta_\ell} = 1+\zeta_\ell+\ldots+\zeta_\ell^{k-1} \ \in \mathcal{O}_\ell$$
are in fact units, that is $\epsilon_k,\epsilon_k^{-1}\in\mathcal{O}_\ell$. Along with the root of unity $\zeta_\ell$, they generate a group of units $\langle\{\zeta_\ell\} \cup \{\epsilon_k \mid k \in [0,\ell)\}\rangle$ called the cyclotomic units. We may rewrite the definition of $q$-binomial coefficients for $r \geq s$
as $$\qbinom{r}{s}_{\zeta_\ell} = \frac{\epsilon_r \epsilon_{r-1}\ldots \epsilon_{r-s+1}}{\epsilon_s \epsilon_{s-1}\ldots \epsilon_1}, $$ which is a cyclotomic unit. Looking for analogies between $TC_{\zeta_\ell}$ and CHS codes, cyclotomic units $\epsilon_k$ play the role of positive integers $k$. From $$\epsilon_k:=\frac{1-\zeta_\ell^k}{1-\zeta_\ell} = \frac{\sin \frac{\pi k}{\ell}}{\sin \frac{\pi}{\ell}}  \zeta_\ell^{-(k-1)(\ell+1)/2},$$ since $k$ is much smaller than $\ell$, $\epsilon_k$ is roughly $k$ times a root of unity. \\ 

\noindent Another interesting observation is that the determinants appearing in the proof of Theorem \ref{tree_codes_cyclotomic_theorem} are nonzero algebraic integers. Trivially, their norm down to the integer has absolute value $1$. To ensure the nonvanishing of the determinants, it thus suffices to represent $\Z[\zeta_\ell]$ in a finite precision big enough to resolve this. The resulting binary tree codes will have the same tradeoff as with the $\{1,\zeta_\ell,\ldots,\zeta_\ell^{\ell-1}\}$ basis representation.

\begin{myrmk}\label{chebotarev_remark}
	Chebotar\"ev showed that for a prime $\ell$ and a primitive $\ell^{th}$ root of unity $\zeta_\ell \in \C$, every square sub-matrix of the Fourier matrix $\{\zeta_\ell^{ij}\mid i,j \in [0,\ell)\}$ has nonzero determinant. His theorem was rediscovered many times, with an account of his original proof in \cite{SL96}. Tao \cite{tao03} gave a new proof and further observed that it implied  additive uncertainty. However the Fourier matrix is not lower triangular, meaning the transformation is not online. Our matrix $\{\qbinom{i}{j}_{\zeta_\ell}\mid i,j \in [0,n)\}$ appearing in the construction of tree codes $TC_{\zeta_\ell}$ appears as the lower triangular matrix in the LU decomposition of the Fourier matrix \cite{OP00}. Theorem \ref{tree_codes_cyclotomic_theorem} could be seen as an online version of Chebotar\"ev's theorem.
\end{myrmk}

\begin{myrmk}\label{hyperinvertible_remark}
 	In cryptography, square matrices all of whose square sub-matrices have nonvanishing determinant are called hyperinvertible. Beerliov\'a-Trub\'iniov\'a and Hirt \cite[\S~3.2]{BH08} constructed $n$ by $n$ hyperinvertible matrices over small finite fields (requiring only $2n$ elements). Hyperinvertibility implies additive uncertainty. It is remarkable that their construction achieves additive uncertainty over such small fields. Their matrices are however not lower triangular. A lower triangular analogue of their construction would yield binary tree codes of constant distance and alphabet size.  
 	Their construction is as follows. Take a field $\F$ with $2n$ distinct elements $\alpha_1,\alpha_2,\ldots,\alpha_n,\beta_1,\beta_2,\ldots,\beta_n$. Consider the linear transformation that takes $(x_1,x_2,\ldots,x_n)\in\F^n$ to $(y_1,y_2,\ldots,y_n)\in\F^n$ as follows. Interpolate a polynomial $g(x)\in\F[x]$ such that $g(\alpha_j)=x_j, \forall j$ and evaluate it at the $\beta$'s to yield $y_i = g(\beta_i), \forall i$. From Lagrange interpolation, the matrix of the transformation takes the form $$M:=\left\{\prod_{k\neq j}\frac{\beta_i-\alpha_k}{\alpha_j-\alpha_k}, i,j\right\}.$$  
  	Take an arbitrary square submatrix $M_{I,J}$ with row and column index sets $I$ and $J$. Being square, $M_{I,J}$ is invertible if the corresponding linear transformation is surjective. We claim this surjection by showing that every target $\{y_i,i \in I\}$ is hit. Equivalently, there is an $(x_1,x_2,\ldots,x_n)$ such that $x_j = 0, \forall j \notin J$ that $M$ maps to a vector that agrees with the target $\{y_i\}$ for $i \in I$. This is evident since there is a degree less than $n$ polynomial $g(x)$ with $g^\prime(\alpha_j)= 0, \forall j \notin J$ and $g^\prime(\beta_i)=y_i, i\in I$. 
 	\end{myrmk}

\section{Tree codes from Transcendence}\label{section_transcendence}
The Lindstr\"om-Gessel-Viennot applied to path lattices plays an important role in Cohen-Haeupler-Schulman codes. We extend the Cohen-Haeupler-Schulman framework by looking to path lattices with weights drawn from carefully chosen transcendental numbers. The first construction is inspired by phyllotaxis \cite{Mit77}, the arrangement of leaves on tree stems. The second construction has several interesting pseudorandom properties, assuming the ABC conjecture. Wrapped in the coding theory framework of \S~\ref{alphabet_reduction_section}, both constructions yield binary tree codes with constant distance and polylogarithmic alphabet size. 

\subsection{Tree codes and Sunflowers}\label{sunflower_subsection}
Let $\theta$ be an irrational algebraic real number and consider the $n$ integer multiples  $$S_{\theta,n}:=\{i\theta \mod 1 \mid i \in [0,n)\}$$ of $\theta$ modulo $1$. Sort and relabel $S_{\theta,n}$ to give $\{s_0,s_1,\ldots,s_{n-1}\}$ where $0=s_0<s_1<\ldots<s_{n-1}$ with $<$ coming from the ordering on the natural lift of the unit interval $\R/\Z$ to $\R$. Steinhaus conjectured that there are at most three gaps and when there are three, the largest gap equals the sum of the other two. That is, when there are three gaps,    $$\{s_i-s_{i-1} \mid i \in [0,n)\} = \{g^{\theta}_{n,1},g^{\theta}_{n,2},g^{\theta}_{n,3}\}$$ where $g^{\theta}_{n,1}\leq g^{\theta}_{n,2} \leq g^{\theta}_{n,3}$ and $g^{\theta}_{n,3} = g^{\theta}_{n,1} + g^{\theta}_{n,2}$. Steinhaus's conjecture was proven by S\'os \cite{sos58} and is commonly referred to as the three gap theorem. It is remarkable that there are only three gaps, irrespective of $n$.\\

\noindent Through the exponential map from $\R/\Z$ to the unit circle, $S_{\theta,n}$ yields a sequence of points on the unit circle. In the previous cyclotomic construction, we substituted for $q$ a root of unity. Now, we will substitute for $q$ the point $e^{2\pi \iota \theta}$ on the unit circle chosen by the angle $\theta$. We use $\iota$ to denote  a square root of $-1$, to not confuse with index $i$. One particular choice for $\theta$, seemingly common in flowering plants (Fibonacci phyllotaxis) \cite{Mit77} is the golden section $\alpha = (\sqrt{5}-1)/2$. van Ravenstein \cite{rav89} proved that this choice $\theta=\alpha$ maximizes the smallest gap $g^\theta_{n,1}$. Further, $g^\alpha_{n,1} = \Omega(1/n)$, meaning the points $\{e^{2 \pi i \iota \alpha} , i \in [0,n)\}$ on the unit circle are well separated. We present our tree codes for arbitrary $\theta$, but our construction is best imagined for the golden section choice $\theta=\alpha$.\\  

\noindent For an irreducible algebraic real $\theta$, from the evaluation map $$z_i  = \sum_{j=0}^{n-1}f_j\qbinom{i}{j}_{e^{2\pi \iota \theta}} , i \in [0,n)$$  
build the tree code $$TC_\theta:\Z^n \longrightarrow (\Z\times \Z[e^{2\pi \iota \theta}])^n \ \ \ \ \ \ \ \ \ \ \ \ \ \ $$
$$\ \ (z_0,z_1,\ldots,z_{n-1})\longmapsto ((z_0,f_0),(z_1,f_1),\ldots,(z_{n-1},f_{n-1})).$$

\begin{mylemma}\label{sunflower_lemma}
	Fix an irrational algebraic real number $\theta$. For every positive integer $n$, the tree code $TC_{\theta}:\Z^n \longrightarrow (\Z\times \Z[e^{2\pi \iota \theta}])^n$ has distance $1/2$. 
\end{mylemma}
\begin{proof}
	Consider nonnegative integers $r_1<r_2<\ldots<r_m$ and $c_1<c_2<\ldots<c_m$ with $n\geq r_i \geq c_i$ for all $i \in [m]$. By Lemma \ref{q-binomial_lemma}, the determinant $$det\left(\left\{\qbinom{r_i}{c_j}_{e^{2\pi \iota \theta}}\mid i,j\in [m]\right\}\right)$$ is a monic integer polynomial in $e^{2\pi \iota \theta}$. By the Lindemann-Weierstrass theorem, $e^{2\pi \iota \theta}$ is transcendental. The  determinant, which is a nontrivial algebraic expression in $e^{2\pi \iota \theta}$, hence does not vanish.  
\end{proof}

\noindent\textbf{Alphabet size bounds:} Again, the alphabet of the binary tree codes resulting from wrapping $TC_{\theta}$ in the coding theory machinery will depend on how the size of the coefficients $f_0,f_1,\ldots,f_{n-1} \in\Z[e^{2\pi \iota \theta}]$ scale. We must choose a concrete representation of $\Z[e^{2\pi \iota \theta}]$ before discussing the size of the coefficients. We will represent the complex numbers involved with a finite precision of $p_n$ bits that scales with $n$. We next show that with precision $p_n$ polynomial in $n$, the tree codes $TC_{\theta}$ still have distance $1/2$. \\

\begin{mythm}\label{sunflower_theorem}
	Fix an irrational algebraic real number $\theta$. For every positive integer $n$, the tree code $TC_{\theta}:\Z^n \longrightarrow (\Z\times \Z[e^{2\pi \iota \theta}])^n$ with coefficients $f_0,f_1,\ldots,f_n$ represented with $\Theta(n^8)$ bit precision has distance $1/2$. 
\end{mythm}
\begin{proof}
	Consider nonnegative integers $r_1<r_2<\ldots<r_m$ and $c_1<c_2<\ldots<c_m$ with $n\geq r_i \geq c_i$ for all $i\in[m]$. By Lemma \ref{q-binomial_lemma}, the determinant $$\det\left(\left\{\qbinom{r_i}{c_j}_{e^{2\pi \iota \theta}}\mid i, j \in [m]\right\}\right)$$ is a monic integer polynomial $$x^{d}+ \sum_{k<d}w_kx^k \in\Z[x]$$ evaluated at $x=e^{2\pi \iota \theta}$ where $d:=\sum_{i=1}^{m}c_i(r_i-c_i) \le n^3$, $w_k \geq 0$ and $$\sum_{k<d}w_k = \det\left(\left\{\binom{r_i}{c_j} \mid i, j \in [m] \right\}\right)-1$$. \\
	
	\noindent By an effective Lindemann-Weierstrass theorem \cite[\S~1]{Ser98} using  Mahler's method \cite{Mah32}, an integer polynomial of bounded degree and height, when evaluated at $e^{2\pi \iota \theta}$, is bounded away from zero by $$\left|e^{2\pi \iota \theta d}+ \sum_{k<d}w_ke^{2\pi \iota \theta k}\right| \geq 2^{-c_\theta d^2\log(h)},$$ where $h \leq \det\left(\left\{\binom{r_i}{c_j} \mid i, j \in [m]\right\}\right)-1$ is the naive height of the polynomial and $c_\theta$ is an absolute positive constant depending only on $\theta$. Note $h = 2^{O(n^2)}$ by Hadamard's identity.  In summary, $\Theta(n^8)$ bits of precision suffice to ensure the determinant does not vanish. It is sufficient to have a precision of the same order for the entries of the matrix, to ensure the determinant does not vanish.
\end{proof}

\noindent The high degree in the polynomial dependence $n^8$ in the bit precision is an artefact of known effective Lindemann-Weierstrass bounds being much weaker than expected. This high degree dependence is not a big concern. When wrapped in the coding theory machinery \S~\ref{alphabet_reduction_section}, we get binary tree codes of constant distance and polylogarithmic alphabet size (see Remark \ref{rmk_larger_alphabet}).

\subsection{Tree codes and the unit circle}\label{unit_circle_subsection}
The distribution of points  $\{i\theta \mod 1 \mid i\in[0,n)\}$ is structured,  satisfying the three gap theorem. This structure is an obstruction to them approaching the uniform distribution, perhaps a desired feature for constructing tree codes amenable to decoding by convex optimization. Weyl \cite{wey16} showed for any irrational algebraic real $\theta$ that the set of integer square multiples $\{i^2\theta \mod 1 \mid i\in[0,n)\}$ converges to the uniform distribution. Rudnick, Sarnak and Zaherescu \cite{RSZ01} related the convergence to the Diaphantine approximability of $\theta$. Assuming the ABC conjecture, they proved a strong convergence theorem. Our hope is that this pseudorandomness can be leveraged to prove the online RIP that is sought in the convex optimization decoding formulation of \S~\ref{convex_decoding_section}.\\

\noindent With the points  $\{i\theta \mod 1\mid i \in [0,n)\}$, associate the set of points $\{x_i:=e^{2\pi\iota i\theta}\mid i\in [0,n)\}$ on the unit circle. For nonnegative integers $r,s$, define
\[ \xbinom{r}{s}_\theta :=  \begin{cases} 
\prod_{k=0}^{s-1}\frac{x_r-x_{s-k-1}}{x_s-x_{s-k-1}} & , r \geq s \\
\ \ \ \ 0 & , r < s 
\end{cases}
\]     
and consider the evaluation map $$z_i  = \sum_{j=0}^{i}g_j\xbinom{i}{j}_{\theta} , i \in [0,n).$$  We refrain from giving an explicit inversion formula, but the evaluation map and its inverse are both online, easily seen by the fact that the inverse of a lower triangular matrix is lower triangular. For nonnegative integers $r_1<r_2,\ldots<r_m$ and $c_1<c_2<\ldots<c_m$ with $r_i \geq c_i$ for all $i \in [m]$, the square sub-matrix  $\{\xbinom{r_i}{c_j}_q\mid i,j\in[m]\}$ of the evaluation map is realised as the path matrix of the following graph.

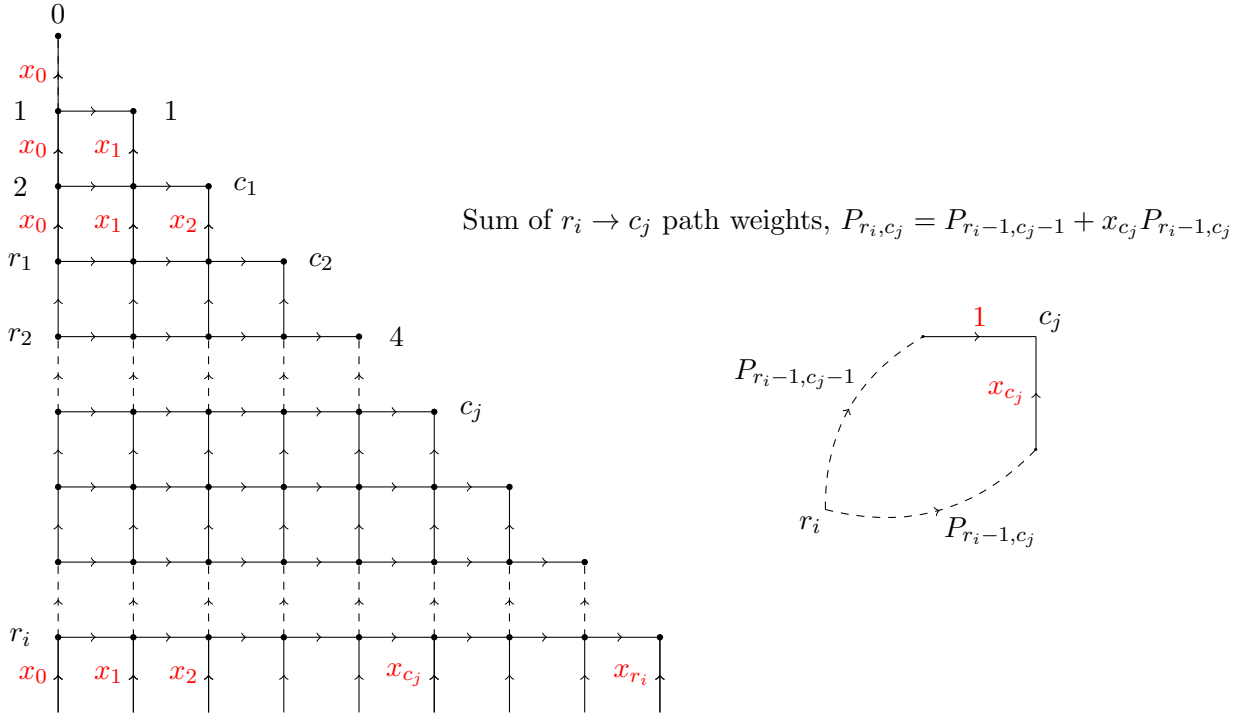
\begin{tikzpicture}
\draw[fill] (0,8) circle (1pt);
\draw[dashed,->-](0,7)--(0,8);
\draw[->-] (0,7) -- (0,8) node[midway,sloped,left,rotate=270] {\textcolor{red}{$x_0$}};
\node at (0, 8.3)  {$0$};

\foreach \y in {1,2,3,4,5,6,7,8}
\foreach \x in {0,...,\y}
\draw[fill] (\x,8-\y) circle (1pt) coordinate (m,\x,\y);

\foreach \y in {1,2,3}
\foreach \x in {0,...,\y}
\draw[->-](\x,7-\y) -- (\x,8-\y);

\foreach \x in {0,1,2,3,4}
\draw[dashed,->-](\x,3)--(\x,4);

\foreach \y in {5,6}
\foreach \x in {0,...,\y}
\draw[->-](\x,7-\y) -- (\x,8-\y);

\foreach \x in {0,1,2,3,4,5,6,7}
\draw[dashed,->-](\x,0)--(\x,1);

\foreach \y in {8}
\foreach \x in {0,...,\y}
\draw[->-](\x,7-\y) -- (\x,8-\y);

\foreach \y in {1,2,3,4,5,6,7,8}
\foreach \x in {1,...,\y}
\draw[->-](\x-1,8-\y) -- (\x,8-\y);

\draw[->-] (0,-1) -- (0,0) node[midway,sloped,left,rotate=270] {\textcolor{red}{$x_0$}};

\draw[->-] (1,-1) -- (1,0) node[midway,sloped,left,rotate=270] {\textcolor{red}{$x_1$}};

\draw[->-] (2,-1) -- (2,0) node[midway,sloped,left,rotate=270] {\textcolor{red}{$x_2$}};

\draw[->-] (0,6) -- (0,7) node[midway,sloped,left,rotate=270] {\textcolor{red}{$x_0$}};

\draw[->-] (1,6) -- (1,7) node[midway,sloped,left,rotate=270] {\textcolor{red}{$x_1$}};

\draw[->-] (0,5) -- (0,6) node[midway,sloped,left,rotate=270] {\textcolor{red}{$x_0$}};

\draw[->-] (1,5) -- (1,6) node[midway,sloped,left,rotate=270] {\textcolor{red}{$x_1$}};

\draw[->-] (2,5) -- (2,6) node[midway,sloped,left,rotate=270] {\textcolor{red}{$x_2$}};

\draw[->-] (5,-1) -- (5,0) node[midway,sloped,left,rotate=270] {\textcolor{red}{$x_{c_j}$}};

\draw[->-] (8,-1) -- (8,0) node[midway,sloped,left,rotate=270] {\textcolor{red}{$x_{r_i}$}};

\node at (-0.5, 7)  {$1$}; 
\node at (-0.5, 6)  {$2$}; 
\node at (-0.5, 5)  {$r_1$}; 
\node at (-0.5, 4)  {$r_2$}; 
\node at (-0.5, 0)  {$r_i$}; 

\node at (1.5,7) {$1$}; 
\node at (2.5,6) {$c_1$}; 
\node at (3.5,5) {$c_2$}; 
\node at (4.5,4) {$4$}; 
\node at (5.5,3) {$c_j$};

\node at (10,1.5) {$r_i$};
\node at (11.5,4) {$.$};
\node at (13.2,4.2) {$c_j$};
\node at (13,2.5) {$.$};
\draw[->-](11.5,4) -- (13,4) node[midway,above]{\textcolor{red}{$1$}};
\draw[->-](13,2.5) -- (13,4) node[midway,sloped,left,rotate=270]{\textcolor{red}{$x_{c_j}$}};
\draw[dashed,->-](10.2,1.7) to[bend left] (11.5,4);
\draw[dashed,->-](10.2,1.7) to[bend right] (13,2.5);

\node at (9.8,3.5) {$P_{r_i-1,c_j-1}$};

\node at (12.4,1.4) {$P_{r_i-1,c_j}$};

\node at (10.5,5.5) {Sum of $r_i \ra c_j$ path weights, $P_{r_i,c_j} = P_{r_i-1,c_j-1}+x_{c_j}P_{r_i-1,c_j}$};

\end{tikzpicture}\\

\noindent By the Lindstr\"om-Gessel-Viennot Lemma, the determinant is a polynomial in $e^{2\pi \iota \theta}$:	$$\det\left(\left\{\xbinom{r_i}{c_j}_q \mid i, j \in [m]\right\}\right) = e^{2\pi \iota \theta d}+ \sum_{k<d}w_k e^{2\pi \iota \theta k}$$ where $d:=\sum_{i=1}^{m}c^2_i(r^2_i-c^2_i)$, $w_k \geq 0$ and $\sum_{k<d}w_k = \det\left(\left\{\binom{r_i}{c_j}\mid i, j \in [m]\right\}\right)-1$. Again, by the Lindemann-Weierstrass lemma, $e^{2\pi \iota \theta}$ is transcendental and the determinant does not vanish.  Hence the tree code $$TC^{\theta^2}:\Z^n \longrightarrow (\Z\times \Z[e^{2\pi\iota \theta}])^n \ \ \ \ \ \ \ \ \ \ \ \ \ \ $$
$$(z_0,z_1,\ldots,z_{n-1})\longmapsto ((z_0,g_0),(z_1,g_1),\ldots,(z_{n-1},g_{n-1}))$$ has distance $1/2$. As in the proof of theorem \ref{sunflower_theorem}, an effective Lindemann-Weierstrass theorem implies that the nonvanishing of the determinant (and hence the distance property) holds even when the precision scales polynomially as $\Theta(n^{11})$. We may hence claim the following theorem.

\begin{mythm}\label{sunflower_theorem2}
	Fix an irrational algebraic real number $\theta$. For every positive integer $n$, the tree code $TC^{\theta^2}:\Z^n \longrightarrow (\Z\times \Z[e^{2\pi \iota \theta}])^n$ with coefficients $g_0,g_1,\ldots,g_n$ represented with $\Theta(n^{11})$ bit precision has distance $1/2$. 
\end{mythm}

\section{Acknowledgements}\label{acknowledgements}
We would like to thank Valerie Berti, Bernhard Haeupler, Antoine Joux and Patrice Philippon for valuable discussions.

\bibliography{general_bib}{}
\bibliographystyle{alpha}

\end{document}